\documentclass[runningheads]{llncs}
\usepackage{datetime}
\usepackage{enumitem}
\usepackage{wrapfig}
\usepackage[T1]{fontenc}
\usepackage[utf8]{inputenc}
\usepackage{pgfplots}
\usepackage{amsmath, amssymb}
\usepackage{color}
\usepackage{cases}
\usepackage{xparse}
\usepackage{xargs}
\usepackage{appendix}
\usepackage[linesnumbered,lined,boxed,commentsnumbered,noend]{algorithm2e}
\usepackage{algpseudocode}
\usepackage{verbatim, xspace}
\usepackage{tikz}
\usetikzlibrary{automata,calc}
\usepackage{mathtools}
\usepackage{tikzit}
\usepackage{float}

\usepackage{hyperref}
\usepackage{pgfplots}

\AtEndEnvironment{proof}{\phantom{}\qed}

\usepackage{xcolor}

\newcommand{\citet}[1]{\cite{#1}}

\usepackage[colorinlistoftodos,textsize=tiny]{todonotes}
\newcommandx{\unsure}[2][1=]{\todo[linecolor=green,backgroundcolor=green!25,bordercolor=green,#1]{\normalsize #2}}
\newcommandx{\improvement}[2][1=]{\todo[inline,linecolor=blue,backgroundcolor=blue!05,bordercolor=blue,#1]{\normalsize #2}}
\newcommandx{\info}[2][1=]{\todo[linecolor=yellow,backgroundcolor=yellow!25,bordercolor=yellow,#1]{#2}}
\newcommandx{\floatmodel}[2][1=]{\todo[inline,linecolor=red,backgroundcolor=yellow!25,bordercolor=yellow,#1]{#2}}
\newcommandx{\thiswillnotshow}[2][1=]{\todo[disable,#1]{#2}}
\newcommandx{\celine}[2][1=]{\todo[inline,linecolor=green,backgroundcolor=green!25,bordercolor=green,caption={\normalsize \textbf{Celine}},#1]{\normalsize #2}}
\newcommandx{\isja}[2][1=]{\todo[inline,linecolor=blue,backgroundcolor=blue!25,bordercolor=blue,caption={\normalsize \textbf{Isja}},#1]{\normalsize #2}}
\newcommandx{\jesper}[2][1=]{\todo[inline,linecolor=red,backgroundcolor=red!25,bordercolor=red,caption={\normalsize \textbf{Jesper}},#1]{\normalsize #2}}
\newcommandx{\krisztina}[2][1=]{\todo[inline,linecolor=gray,backgroundcolor=red!25,bordercolor=red,caption={\normalsize \textbf{Krisztina}},#1]{\normalsize #2}}
\newtheorem{numclaim}{Claim}
\newcommand{\qeds}{$\blacksquare$}
\newenvironment{subproof}{\textit{Proof of Claim.}}{\qeds}

\newcommand{\Oh}{\mathcal{O}}
\newcommand{\Os}{\Oh^{\star}}

\newcommand{\N}{\mathbb{N}}

\newcommand{\overl}{\overleftarrow}
\newcommand{\overr}{\overrightarrow}

\DeclareMathOperator\supp{supp}

\newcommand{\pw}{\ensuremath{pw}}
\newcommand{\tw}{\ensuremath{tw}}

\renewcommand{\leq}{\leqslant}
\renewcommand{\geq}{\geqslant}
\renewcommand{\le}{\leqslant}
\renewcommand{\ge}{\geqslant}
\newcommand{\In}{\mathsf{in}}
\newcommand{\Out}{\mathsf{out}}
\newcommand{\Ca}{\mathsf{cap}}
\newcommand{\Dem}{\mathsf{dem}}
\newcommand{\Dis}{\mathsf{cost}}

\newcommand{\defproblem}[3]{
	\vspace{2mm}
	\vspace{1mm}
	\noindent\fbox{
		\begin{minipage}{0.95\textwidth}
			#1 \\
			{\bf{Input:}} #2  \\
			{\bf{Task:}} #3
		\end{minipage}
	}
	\vspace{2mm}
}

\title{On the Parameterized Complexity of the Connected Flow and Many Visits TSP Problem\thanks{Supported by the project CRACKNP that has received funding from the European Research Council (ERC) under the European Union’s Horizon 2020 research and innovation programme (grant agreement No 853234) and by the Netherlands Organization for Scientific Research under project no. 613.009.031b.}}
\titlerunning{Parameterized Complexity of Connected Flow and Many-visits TSP}

\author{
	Isja Mannens\inst{1} %\thanks{Utrecht University, The
	\and	
    Jesper Nederlof\inst{1}\orcidID{0000-0003-1848-0076} %\footnotemark[1]
    \and
    Céline Swennenhuis\inst{2}\orcidID{0000-0001-9654-8094} %\thanks{Eindhoven University of Technology, The
	\and  
	Krisztina Szil\'agyi\inst{1} %\footnotemark[1]
}

\institute{Utrecht University, The
	Netherlands,\\ \texttt{\{i.m.e.mannens, j.nederlof, k.szilagyi\}@uu.nl}.\and
	Eindhoven University of Technology, The Netherlands, \texttt{c.m.f.swennenhuis@tue.nl}. }

\begin{document}

\maketitle

\thispagestyle{empty}
We study a variant of \textsc{Min Cost Flow} in which the flow needs to be connected. Specifically, in the \textsc{Connected Flow} problem one is given a directed graph $G$, along with a set of demand vertices $D \subseteq V(G)$ with demands $\Dem: D \rightarrow \mathbb{N}$, and costs and capacities for each edge. The goal is to find a minimum cost flow that satisfies the demands, respects the capacities and induces a (strongly) connected subgraph. This generalizes previously studied problems like the \textsc{(Many Visits) TSP}.

We study the parameterized complexity of \textsc{Connected Flow} parameterized by $|D|$, the treewidth $\tw$ and by vertex cover size $k$ of $G$ and provide:
\begin{enumerate}
	\item \textsf{NP}-completeness already for the case $|D|=2$ with only unit demands and capacities and no edge costs, and fixed-parameter tractability if there are no capacities,
	
	\item a fixed-parameter tractable $\Os(k^{\Oh(k)})$ time algorithm for the general case, and a kernel of size polynomial in $k$ for the special case of \textsc{Many Visits TSP},
	
	\item a $|V(G)|^{\Oh(\tw)}$ time algorithm and a matching $|V(G)|^{o(\tw)}$ time conditional lower bound conditioned on the Exponential Time Hypothesis.
\end{enumerate}
To achieve some of our results, we significantly extend an approach by Kowalik et al.~[ESA'20].

\begin{picture}(0,0)
\put(462,-270)
{\hbox{\includegraphics[width=40px]{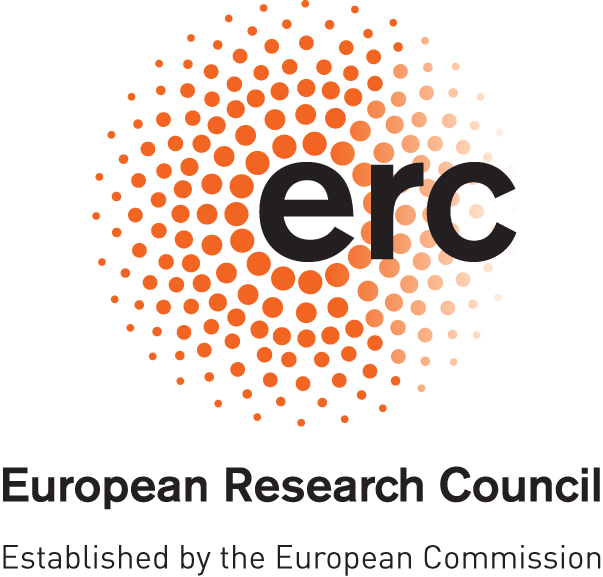}}}
\put(452,-330)
{\hbox{\includegraphics[width=60px]{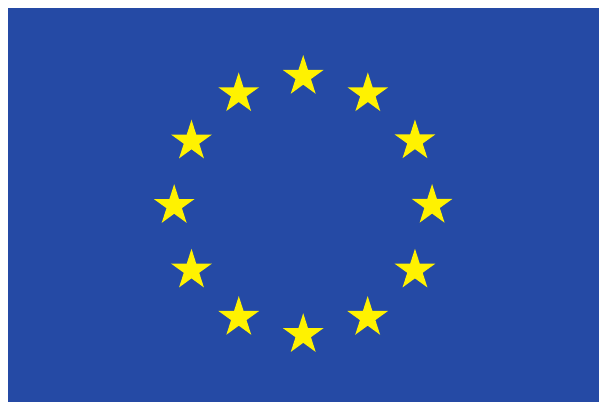}}}
\end{picture}

%\clearpage
%\setcounter{page}{1}

\section{Introduction}
In the \textsc{Connected Flow} problem we are given a directed graph $G=(V,E)$ with costs and capacities on the edges and a set $D \subseteq V$ such that each $v\in D$ has a fixed demand. We then ask for a minimum cost \emph{connected flow} on the edges that satisfies the demand for each $v \in D$, i.e. we look for a minimum cost \emph{flow conserving} function $f:E\to \N$, such that the set of edges with strictly positive flow $f$ is connected and the total flow coming into $v \in D$ is equal to its demand (see below for a formal definition of the problem).

One arrives (almost) directly at the \textsc{Connected Flow} problem by adding a natural connectivity constraint to the well known \textsc{Min Cost Flow} problem (from now on abbreviated with simply `\textsc{Flow}', see Appendix~\ref{sec:problemdefs} for details). But unfortunately, \textsc{Connected Flow} has the same fate as many other slight generalizations of \textsc{Flow}: The additional requirement changes the complexity of the problem from being solvable in polynomial time to being \textsf{NP}-complete (see~\cite[Section A2.4]{GareyJ79} for more of such \textsf{NP}-complete generalizations).

The problem generalizes a number of problems, including the \textsc{Many Visits TSP (MVTSP)}\footnote{In this problem a minimum length tour is sought that satisfies each vertex a given number of times. The generalization is by setting the demand of a vertex to the number of times the tour is required to visit that vertex and using infinite capacities.}. This problem has a variety of potential applications in scheduling and computational geometry (see e.g the discussion by Berger et al.~\cite{BergerKMV20}), and its study from the exponential time perspective recently witnessed several exciting results. In particular, Berger et al.~\cite{BergerKMV20} improved an old $n^{\Oh(n)}$ time algorithm by Cosmadakis and Papadimitriou~\cite{CosmadakisP84} to $\Os(5^n)$ time and polynomial space, and recently the analysis of that algorithm was further improved by Kowalik et al.~\cite{kowalik_et_al:LIPIcs:2020:12932} to $\Os(4^n)$ time.

The \textsc{Connected Flow} problem also generalizes other problems studied in parameterized complexity, such as the \textsc{Eulerian Steiner Subgraph} problem, that was used in an algorithm for \textsc{Hamiltonian Index} by Philip et al.~\cite{PhilipRS20}, or the problem of finding $2$ short edge
disjoint paths in undirected graphs (whose parameterized complexity was for example studied by Cai and Ye~\cite{CaiY16}).

Based on these connections with existing literature on in particular the \textsc{MVTSP}, its appealing formulation, and it being a direct extension of the well-studied \textsc{Flow} problem, we initiate the study of the parameterized complexity of \textsc{Connected Flow} in this paper.

\paragraph*{Our Contributions.}
We first study the (arguably) most natural parameterization: the number of demand vertices for which we require a certain amount of flow.
We show that the problem is \textsf{NP}-complete even in a very special case:
\begin{theorem}\label{thm:NP}
	\textsc{Connected Flow} with $2$ demand vertices is \textsf{NP}-complete.
\end{theorem}

The reduction heavily relies on the capacities and we show that this is indeed what makes the problem hard. Namely, using the algorithm for \textsc{MVTSP} from \cite{kowalik_et_al:LIPIcs:2020:12932}, we get an algorithm that can solve instances of \textsc{Connected Flow} if all capacities are infinite:

\begin{theorem}\label{thm:inf_cap}
	Any instance instance $(G, D, \Dem, \Dis, \Ca)$ of \textsc{Connected Flow} where $\Ca(e)=\infty$ for all $e\in E$ can be solved in time $\Os(4^{|D|})$.
\end{theorem}

Next we study a typically much larger parameterization, the size $k$ of a vertex cover of $G$. One of our main technical contributions is that \textsc{Connected Flow} is fixed-parameter tractable, parameterized by $k$:
\begin{theorem} \label{thm:XFPT}
	There is an algorithm solving a given instance $(G,D,\Dem,\Dis,\Ca)$ of \textsc{Connected Flow} such that $G$ has a vertex cover of size $k$ in time $\Os(k^{\Oh(k)})$.
\end{theorem}

Theorem~\ref{thm:XFPT} is interesting even for the special case of~\textsc{MVTSP} as it generalizes the $\Os(n^n)$ time algorithm from Cosmadakis and Papadimitriou~\cite{CosmadakisP84}, though it is a bit slower than the
more recent algorithms from~\cite{BergerKMV20,kowalik_et_al:LIPIcs:2020:12932}.
For this special case, we even find a polynomial kernel:
\begin{theorem} \label{thm:Xkernel}
	\textsc{MVTSP} admits a kernel polynomial in the size $k$ of the vertex cover of $G$. 
\end{theorem}
The starting point of the proofs of both Theorem~\ref{thm:XFPT} and Theorem~\ref{thm:Xkernel} is a strengthening of a non-trivial lemma from Kowalik et al.~\cite{kowalik_et_al:LIPIcs:2020:12932} which proves the existence of a solution $s'$ that is `close' to a solution $r$ of the \textsc{Flow} problem instance obtained by relaxing the connectivity requirement. 
Since such an $r$ can be found in polynomial time, it can be used to determine how the optimal solution roughly looks.

This is subsequently used by a dynamic programming algorithm that aims to find such a solution close to $r$ to establish Theorem~\ref{thm:XFPT}; the restriction to solutions being close to $r$ crucially allows us to evaluate only $\Os(k^{\Oh(k)})$ table entries.
Additionally this is used in the kernelization algorithm of Theorem~\ref{thm:Xkernel} to locate a set of $\Oh(k^5)$ vertices such that only edges incident to vertices in this set will have a different flow in $r$ and $s'$.

The last parameter we consider is the \emph{treewidth}, denoted by $\tw$, of $G$, which is a parameter that is widely used for many graph problems and that is smaller than $k$.
We present a Dynamic Programming algorithm for \textsc{Connected Flow}:

\begin{theorem}\label{thm:TwDP}
	Let $M$ be an upper bound on the demands in the input graph $G$, and suppose a tree decomposition of width $\tw$ of $G$ is given. Then a \textsc{Connected Flow} instance with $G$ can be solved in time $|V(G)|^{\Oh(\tw)}$ and an \textsc{MVTSP} instance with $G$ can be solved in time $\min\{|V(G)|,M\}^{\Oh(\tw)}|V(G)|^{\Oh(1)}$.
\end{theorem}

We also give a matching lower bound for \textsc{MVTSP}. This lower bound heavily builds on previous approaches, and in particular, some gadgets from Cygan et al.~\cite{DBLP:journals/corr/abs-1211-1506}.

\begin{theorem} \label{thm:tw=hard}
	Assuming the Exponential Time Hypothesis, \textsc{MVTSP} cannot be solved in time $f(\tw)|V(G)|^{o(\tw)}$ for any computable function $f(\cdot)$. 
\end{theorem}

Note that since \textsc{MVTSP} is a special case of \textsc{Connected Flow} this lower bound extends to \textsc{Connected Flow}.

\paragraph*{Notation and Formal Problem Definitions.}
We let $\Os(\cdot)$ omit factors polynomial in the input size. We assume that all integers are represented in binary, so in this paper the input size will be polynomial in the number of vertices of the input graph and the logarithm of the maximum input integer. For a Boolean $b$ we define $[b]$ to be $1$ if $b$ is true and $0$ otherwise. For integers $a$ and $b$ we denote $[a,b]$ as the set of all integers $i$ such that $a \le i \le b$.
All graphs in this paper are directed unless stated otherwise. 

We use the notion of \emph{multisets}, which are sets in which the same element may appear multiple times. Formally, a multiset is an ordered pair $(A, m_A)$ consisting of a set $A$ and a multiplicity function $m_A:A\rightarrow \mathbb{Z}^+$. We slightly abuse notation and let $m_A(e)=0$ if $e\not \in A$.
We can see flow $f$ as a multiset of directed edges, where each edge appears $f(e)$ number of times. We then say that $f(e)$ is the \emph{multiplicity} of $e$.
Given a function $f:E\rightarrow \N$, we define $G_f=(V',E')$  as the multigraph where $e\in E'$ has multiplicity $f(e)$ and $V'$ is the set of vertices incident to at least one $e\in E'$. We let $E(G_f)$ be equal to the multiset $E'$. We also define $\supp(f)=\{e\in E: f(e)>0\}$ as the \emph{support} of $f$. 
The formal statement of \textsc{Connected Flow} is as follows:

\defproblem{ \textsc{Connected Flow}}{$G = (V,E)$, $D \subseteq V$, $\Dem : D \to \N$, $\Dis: E \to \N$, $\Ca: E \to \N\cup\{\infty \}$}{Find a function $f : E \to \N$ such that
	\begin{itemize}	[noitemsep,topsep=2pt]
		\item $G_f$ is connected, 
		\item for every $v \in V$ we have $\sum_{(u,v)\in E} f(u,v)  = \sum_{(v,u)\in E} f(v,u)$,
		\item for every $v \in D$ we have 
		$\sum_{(u,v)\in E} f(u,v) = \Dem(v)$,
		\item for every $e \in E: f(e) \le \Ca(e)$,
	\end{itemize}
	and the value $\Dis(f) = \sum_{e \in E} \Dis(e)f(e)$ is minimized.}

Note that $G_f$ in the above definition is Eulerian (every vertex has the same in and out degree), so it is strongly connected if and only if it is weakly connected. 
We define \textsc{Flow} as the \textsc{Connected Flow} problem without the connectivity requirement, which can be solved in polynomial time\footnote{In Appendix~\ref{sec:problemdefs} we show that \textsc{Flow} is equivalent to the \textsc{Min Cost Flow} problem, which is polynomial-time solvable.}. 
\textsc{MVTSP} is a special case of \textsc{Connected Flow}, where $D=V$ and capacities are infinite. Formal definitions of these problems can be found in Appendix~\ref{sec:problemdefs}.

\paragraph*{Organization.}
The remainder of this paper is organized as follows: in Section~\ref{sec:appdx} we study the parameterization by the number of demand vertices. We show \textsf{NP}-completeness and discuss the reduction of the infinite capacities case of \textsc{Connected Flow} to \textsc{MVTSP}.

In Section~\ref{sec:vert_cover} we first introduce an extension of a lemma from Kowalik et al.~\cite{kowalik_et_al:LIPIcs:2020:12932} that shows that we can transform an optimal solution to the \textsc{Flow} relaxation
to include a specific edge set from an optimal solution of the original \textsc{Connected Flow} instance, without changing too many edges. 
This lemma is subsequently used in Section~\ref{sec:dp} to prove Theorem~\ref{thm:XFPT} and in Section~\ref{sec:polyker} to prove Theorem~\ref{thm:Xkernel}.

In Section~\ref{sec:pathw} we discuss the parameterization by treewidth and pathwidth, giving a Dynamic Programming algorithm for \textsc{Connected Flow} and a matching lower bound for \textsc{MVTSP}. 

We conclude the paper with a discussion on further research opportunities. In Appendix~\ref{app:Definitions} we provide formal problem definitions for \textsc{Flow} and prove it is equivalent to \textsc{Min Cost Flow}. 

\section{Parameterization by number of demand vertices}\label{sec:appdx}\label{app:reductionD}\label{app:redtoMVTSP}
In this section we study the parameterized complexity of \textsc{Connected Flow} with parameter $|D|$, the number of vertices with a demand. We first prove that the problem is \textsf{NP}-hard, even for $|D|=2$, by a reduction from the problem of finding two vertex disjoint paths in a directed graph. Next we show that, if $\Ca(e) = \infty$ for all $e\in E$, the problem can be reduced to an instance of \textsc{MVTSP}, and hence solved in time $\Os(4^{|D|})$.

\begingroup
\def\thetheorem{\ref{thm:NP}}
\begin{theorem}
	\textsc{Connected Flow} with $2$ demand vertices is \textsf{NP}-complete.
\end{theorem}
\addtocounter{theorem}{-1}
\endgroup

\begin{proof} % reduction from vertex disjoint paths
	We give a reduction from the problem of finding two vertex-disjoint paths in a directed graph to \textsc{Connected Flow} with demand set $D$ of size 2. The directed vertex-disjoint paths problem has been shown to be \textsf{NP}-hard for fixed $k = 2$ by Fortune et al.~\cite{FortuneHW80}, so this reduction will prove our theorem for $|D| = 2$. Note that the case of $|D| > 2$ is harder, since we can view $|D| = 2$ as a special case, by adding isolated vertices with demand 0.
	
	Given a graph $G$ and pairs $(s_1, t_1)$ and $(s_2, t_2)$, we construct an instance $(G', D, \Dem, \Dis, \Ca)$ of \textsc{Connected Flow}. Let $V_0 = V \setminus \{s_1, s_2, t_1, t_2\}$, we define
	\[
	V(G') = \{s_1, s_2, t_1, t_2\} \cup \{v_{\In} : v \in V_0\} \cup \{v_{\Out} : v \in V_0\}
	\]
	We let $D = \{s_1, s_2\}$ and set $\Dem(s_1) = \Dem(s_2) =1$. We also define
	\begin{align*}
		E(G') = 	&\{(v_{\In}, v_{\Out}) : v \in V_0\} \\
		&\cup \{(s_i, v_{\In}) : (s_i, v) \in E(G), i = 1, 2 \} \\
		&\cup \{(v_{\Out}, t_i) : (v, t_i) \in E(G), i = 1, 2\} \\
		&\cup \{(u_{\Out}, v_{\In}): u, v \in V_0, (u,v) \in E(G)\} \\
		&\cup \{(t_1, s_2), (t_2, s_1)\}.
	\end{align*}
	We now set $\Dis(u,v) = 0$ and $\Ca(u,v)=1$ for every $(u,v) \in E(G')$. We prove that $G$ has two vertex-disjoint paths (from $s_1$ to $t_1$ and from $s_2$ to $t_2$) if and only if $(G',D,\Dem,\Dis,\Ca)$ has a connected flow of cost $0$.
	
	Let $P_1$ and $P_2$ be two vertex disjoint paths in $G$, from $s_1$ to $t_1$ and from $s_2$ to $t_2$ respectively. Intuitively we will simply walk through the same two paths in $G'$ and then connect the end of one to the start of the other. More formally, we construct a flow $f$ in $G'$ as follows. Let $P_1 = s_1, v^1, \dots, v^\ell, t_1$, we set $f(s_1, v_{\In}^1) = f(v_{\Out}^\ell, t_1) = 1$ as well as $f(v_{\In}^i, v_{\Out}^i) = f(v_{\Out}^i, v_{\In}^{i+1}) = 1$ for all $i\in[1,\ell]$. We do the same for $P_2$. Finally we set $f(t_1, s_2) = f(t_2, s_1) = 1$ and set $f$ to $0$ for all other edges. We note that all capacities have been respected and all demands have been met. The resulting flow is connected, since the paths were connected and $f(t_1, s_2) = 1$.
	
	For the other direction, let $f$ be a connected flow for $(G', D, \Dem, \Dis, \Ca)$. Since $\Dem(s_1) = \Dem(s_2) = 1$ and $s_1$ and $s_2$ only have one incoming edge, we have that $f(t_1, s_2) = f(t_2, s_1) = 1$. We argue that $G_f - \{(t_1, s_2), (t_2, s_1)\}$ consists of two vertex disjoint paths in $G'$, one from $s_1$ to $t_1$ and the other from $s_2$ to $t_2$. First we note that for every vertex in $G'$, it has either in-degree $1$ or out-degree $1$ (or possibly both). This means that since we have $\Ca(u,v) = 1$ for every $(u,v) \in E(G')$, every vertex in $V(G_f)$ has in- and out-degree $1$ in $G_f$. Since $G_f$ is connected we find that $G_f$ is a single cycle and thus $G_f - \{(t_1, s_2), (t_2, s_1)\}$ is the union of two vertex-disjoint paths. We now find two vertex-disjoint paths in $G$ by contracting the edges $(v_{\In}, v_{\Out})$ in $G_f - \{(t_1, s_2), (t_2, s_1)\}$.
\end{proof}

\begin{lemma}\label{lem:remove_vert}
	Given an instance $(G, D, \Dem, \Dis, \Ca)$ of \textsc{Connected Flow} where $\Ca(e)=\infty$ for all $e\in E$, we can construct an equivalent instance of \textsc{MVTSP} on $|D|$ vertices.
\end{lemma}
\begin{proof}
	We construct an equivalent instance $(G', \Dem, \Dis')$ of \textsc{MVTSP} as follows. First we let $V(G') = D$ and for $u,v \in D$ we let $(u,v) \in E(G')$ if and only if there is a $u-v$ path in G, disjoint from other vertices in $D$. We then set $\Dis(u,v)$ to be the total cost of the shortest such path. We keep $\Dem(v)$ the same. 
	
	We now show equivalence of the two instances. Let $s' : E(G') \to \N$ be a valid tour on $(G', \Dem, \Dis')$. We construct a connected flow $f$ on $(G, D, \Dem, \Dis, \Ca)$ by, for each $(u,v) \in E(G')$ adding $s'(u,v)$ copies of the shortest $D$-disjoint $u$-$v$-path in $G$ to the flow. Note that the demands are met, since the demands in both instances are the same. Also note that by definition the total cost of $s'(u,v)$ copies of the shortest $D$-disjoint $u-v$ path is equal to $s'(u,v)\cdot \Dis'(u,v)$ and thus the total cost of $f$ is equal to that of $s'$. Finally we note that the capacities are trivially met.
	
	In the other direction, let $f : E(G) \to \N $ be an optimal connected flow on $(G, D, \Dem, \Dis, \Ca)$. Note that $G_f$ is connected and that every vertex in this multigraph has equal in- and out-degrees. This means we can find some Eulerian tour on $G_f$. 
	We now construct an \textsc{MVTSP} tour $s'$ on $G'$ by adding the edge $(u,v)$ every time $v$ is the first vertex with demand to appear after an appearance of $u$ in the Eulerian tour. Again it is easy to see that $s'$ is connected and that the demands are met. The total cost of $s'$ is the same as $f$, namely if it is larger, then there is some pair $u,v \in D$ such that the cost of some path in the Eulerian tour from $u$ to $v$ is less than $\Dis'(u,v)$, which contradicts the definition of $\Dis'$. If it were smaller, then there is some $D$-disjoint path in the Eulerian tour from some $u$ to some $v$ which is longer than $\Dis'(u,v)$. We can then find a cheaper flow by replacing this path with the shortest path, contradicting the optimality of $f$.
\end{proof}

Since \textsc{MVTSP} can be solved in $\Os(4^n)$ time by Kowalik et al.~\cite{kowalik_et_al:LIPIcs:2020:12932}, we get as a direct consequence:

\begingroup
\def\thetheorem{\ref{thm:inf_cap}}
\begin{theorem}
	Any instance instance $(G, D, \Dem, \Dis, \Ca)$ of \textsc{Connected Flow} where $\Ca(e)=\infty$ for all $e\in E(G)$ can be solved in time $\Os(4^{|D|})$.
\end{theorem}	
\addtocounter{theorem}{-1}
\endgroup

\section{Parameterization by vertex cover}\label{sec:vert_cover}
In this section, we consider \textsc{Connected Flow} and \textsc{MVTSP}, parameterized by the cardinality $k$ of a vertex cover of the input graph. We first extend a lemma from Kowalik et al.~\cite{kowalik_et_al:LIPIcs:2020:12932} to instances of \textsc{Connected Flow}. Then we use this lemma to obtain a fixed-parameter tractable algorithm for \textsc{Connected Flow} and a polynomial-sized kernel for \textsc{MVTSP}.
 
\subsection{Transforming the flow relaxation to enforce some edges}\label{sec:redlemma}
Let $s$ be an optimal solution of \textsc{Connected Flow} and let $T\subseteq \supp(s)$.  We prove that, given any optimal solution $r$ for \textsc{Flow}, there is always a flow $f$ that is close to $r$ and $T\subseteq \supp(f)$. Furthermore it has cost $\Dis(f) \le \Dis(s)$. Note that if $T$ connects all demand vertices to each other, this implies that $f$ is connected and thus an optimal solution of \textsc{Connected Flow}. 

The basic idea and arguments are from Kowalik et al.~\cite{kowalik_et_al:LIPIcs:2020:12932}, where a similar theorem for \textsc{MVTSP} was proved. We adjusted their proof to the case with capacities and where not all vertices have a demand. Furthermore, we noted that we can restrict the tours $C\in \mathcal{C}$ in the proof to be inclusion-wise minimal, which allows us to conclude a stronger inequality.   

\begin{lemma}\label{lem:s-rtours}
	Fix an input instance $(G,D,\Dem,\Dis,\Ca)$ with $G= (V,E)$. Let $s$ be an optimal solution of \textsc{Connected Flow} and let $T\subseteq \supp(s)$. For every optimal solution $r$ of \textsc{Flow}, there is a flow $f$ with $\Dis(f)\le \Dis(s)$, with $f(e)>0$ for all $e\in T$ and such that for every $v \in V$:
		$$\sum_{u \in V} |r(u,v) - f(u,v)| \le 2|T|, \qquad \text{ and } \qquad \sum_{u \in V} |r(v,u) - f(v,u)| \le 2|T|. $$ 
\end{lemma}

\begin{proof}
	We follow the structure of the proof of Lemma 3.2 from Kowalik et al.~\cite{kowalik_et_al:LIPIcs:2020:12932}.  
	We build a flow $f$ (not necessarily optimal for \textsc{Flow}), containing $T$ and with multiplicities close to $r$.
	Recall that $m_B$ denotes the multiplicity function of the multiset $B$.
	We define the multisets of edges $A_s$, $A_r$ and $A$ such that for all $e\in E$:
	\begin{itemize}[noitemsep,topsep=2pt]
		\item $m_{A_s}(e) = \max \{s(e)-r(e),0\}$, 
		\item $m_{A_r}(e) = \max\{r(e)-s(e),0\}$, and
		\item $m_A(e) = \max\{m_{A_r}(e),m_{A_s}(e)\} = \max\{s(e)-r(e), r(e)-s(e)\}$.
	\end{itemize} 
	
	Note that $A$ is the symmetric difference of $s$ and $r$, and therefore any $e \in A$, is exactly either in $A_r$ or in $A_s$, but never in both.
	
	Let $H$ be a tour (i.e. a closed walk) of undirected edges. We then say that $\overr{H}$ is a \emph{cyclic orientation} of $H$ if it is an orientation of the edges in $H$ such that $\overr{H}$ forms a directed tour. A directed edge $e$ that overlaps with $H$ is in \emph{positive orientation} with respect to $\overr{H}$ if it has the same orientation, and negative otherwise. We now define $(s-r)$ directed tours, of which an example is shown in Figure~\ref{fig:srdirected}. 
	\begin{definition}
		Let $C=(e_0,\dots,e_{\ell})\subseteq A$ be a set of edges such that its underlying undirected edge set $H$ is a tour. We then say that $C$ is an \emph{$(s-r)$ directed tour} if there is an orientation $\overr{H}$ of $H$ such that:
		  	\begin{itemize}[noitemsep,topsep=2pt]
		  	\item if $e \in C$ is in positive orientation with respect to $\overr{H}$, then $e \in A_s$,
		  	\item if $e \in C$ is in negative orientation with respect to $\overr{H}$, then $e \in A_r$,
		  	\item if two subsequent edges $e_{i},e_{i+1}$ of $C$ have the same orientation, then their shared vertex, $v$, is not in $D$. This also holds for the edge pair $(e_\ell,e_0)$.
		  \end{itemize} 
	\end{definition}
	
	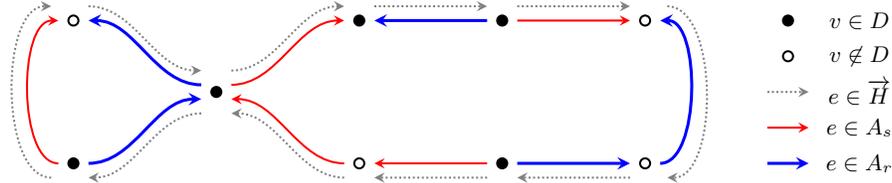
\begin{figure}[ht!]
		\centering
		\scalebox{0.95}{\begin{tikzpicture}[stealth-stealth, shorten >=6pt, shorten <=6pt,thick ]

	%\node[] at (0,9.5) {\Large $X$};
	
	%\draw[blue] (0,0) to[out = -0,in = -140] (2,0.5);
	%\draw[blue] (2,0.5) to[out = 130,in = -0] (0,1);
	%\draw[blue,rounded corners= 30pt]  (0,0)-- (2.6,0.5) --  (0,1);

	%\draw[blue] (0,4) to[ out = 0, in = -145] (2,4);
	%\draw[blue] (0,5) to[ out = 0, in =145] (2,5);
	%\draw[blue] (2,5) to[out = 0, in = 0] (2,4);
	
	\filldraw[black] (0,0) circle (2pt);
	\draw[black] (-2,1) circle (2pt);
	\filldraw[black] (-2,-1) circle (2pt);
	
	\filldraw[black] (2,1) circle (2pt);
	\draw[black] (2,-1) circle (2pt);
	\filldraw[black] (4,1) circle (2pt);
	\filldraw[black] (4,-1) circle (2pt);
	\draw[black] (6,1) circle (2pt);
	\draw[black] (6,-1) circle (2pt);
	
	\draw[blue, very thick, -stealth] (0,0.1) to[out = 180, in = 0] (-2,1);
	\draw[red, -stealth] (-2,-1) to[out = 180, in = 180] (-2,1);
	\draw[blue, very thick, -stealth] (-2,-1) to[out = 0, in = 180] (0,-0.1);
	\draw[red, -stealth] (2,-1) to[out =180, in = 0] (0,-0.1);
	\draw[red, -stealth] (4,-1) to[out =180, in = 0] (2,-1);
	\draw[blue, very thick, -stealth] (4,-1) to[out =0, in = 180] (6,-1);
	\draw[blue, very thick, -stealth] (6,-1) to[out =0, in = 0] (6,1);
	\draw[red, -stealth] (4,1) to[out = 0, in = 180] (6,1);
	\draw[blue, very thick, -stealth] (4,1) to[out =180, in = 0] (2,1);
	\draw[red, -stealth] (0,0.1) to[out =0, in = 180] (2,1);
	
	\draw[gray, densely dotted, -stealth] (-2,1.2) to[out = 0, in = 180]  (0,0.3);
	\draw[gray,densely dotted, -stealth] (-2.1,-1.2) to[out = 180, in = 180] (-2.1,1.2);
	\draw[gray,densely dotted, -stealth] (0,-0.3) to[out = 180, in = 0] (-2,-1.2);
	\draw[gray,densely dotted, -stealth] (2,-1.2) to[out =180, in = 0] (0,-0.3);
	\draw[gray,densely dotted, -stealth] (4,-1.2) to[out =180, in = 0] (2,-1.2);
	\draw[gray,densely dotted, -stealth] (6,-1.2) to[out = 180, in = 0] (4,-1.2);
	\draw[gray,densely dotted, -stealth] (6.1,1.2) to[out =0, in = 0] (6.1,-1.2);
	\draw[gray,densely dotted, -stealth] (4,1.2) to[out = 0, in = 180] (6,1.2);
	\draw[gray,densely dotted, -stealth] (2,1.2) to[out =0, in = 180] (4,1.2);
	\draw[gray,densely dotted, -stealth] (0,0.3) to[out =0, in = 180] (2,1.2);

	\filldraw[black] (8,1) circle (2pt);
	\draw[black] (8,.5) circle (2pt);
	\node[] at (9,1) {$v \in D$};
	\node[] at (9,.5) {$v \not\in D$};
	
	\draw[gray,densely dotted, -stealth] (7.5,0) -- (8.5,0);
	\draw[blue, very thick, -stealth] (7.5,-1) -- (8.5,-1);
	\draw[red, -stealth] (7.5,-.5) -- (8.5,-.5);
	\node[] at (9,-1) {$e \in A_r$};
	\node[] at (9,-.5) {$e \in A_s$};
	\node[] at (9,0) {$e \in \overr{H}$};
	%\draw (7.5,-1) rectangle (10,1.5);
	
\end{tikzpicture}}
		\caption{Example of an $(s-r)$ directed tour. Note that every time the tour visits a vertex $v \in D$, the orientation changes. However if the tour visits a vertex $v \not \in D$, the orientation might not change.}
		\label{fig:srdirected}
	\end{figure}
	
	We give a construction such that $A$ can be partitioned  into a multiset of $(s-r)$ directed tours. We take $(u,v)\in A$ arbitrarily as our first edge of our walk and iteratively add edges until we find an $(s-r)$ directed tour. We assume the current edge $(u,v)$ is in $A_s$ (if the edge is in $A_r$, the arguments are similar).   
	If $v\in D$, then there exists $(v,w) \in A_r$, because $v$ is visited $\Dem(v)$ times by both $r$ and $s$. If $v\not \in D$, there exists either $(v,w)\in A_r$ or $(w,v)\in A_s$ because $A$ is the symmetric different of the flows $r$ and $s$. We take this edge as the next edge in our $(s-r)$ directed tour. This way we can keep finding the next edges, until we can take our first edge $(u,v)$ as our next edge and we find an $(s-r)$ directed tour. 
	We then remove this tour and inductively find the next until $A$ is empty.
		
	It follows that $A$ can be partitioned into a multiset $\mathcal{C}$ of $(s-r)$ directed tours, i.e. 
	$$ m_A = \sum_{C\in \mathcal{C}} m_C,$$ where $m_C$ is the multiplicity of $(s-r)$ directed tour $C$. 

	We may assume that these $(s-r)$ directed tours are inclusion-wise minimal, i.e. for each $(s-r)$ directed tour $C \in \mathcal{C}$, no subset $C' \subset C$ is an $(s-r)$ directed tour. 
	Otherwise, $C$ can be split into two disjoint $(s-r)$ directed tours $C'$ and $C\setminus C'$. 

	\begin{numclaim}
	For any $v\in V$ and any inclusion-wise minimal $C\in\mathcal{C}$:
	\begin{equation}\label{eq:$(s-r)$ directedtour}
		\sum_{u\in V}[(u,v) \in C] \le 2 \qquad \text{ and } \qquad  \sum_{u\in V}[(v,u) \in C] \le 2.
	\end{equation}	
	\end{numclaim}
\begin{subproof}
	We only prove the first inequality. The second inequality can be proved with an analogous argumentation.
	Assume not, i.e. assume there exists $C\in\mathcal{C}$ and $v\in V$ such that there exist $x_1,x_2,x_3 \in V$ with $(x_i, v) \in C$ for $i=1,2,3$. Each of these edges must be either in $A_s$ or $A_r$. Assume without loss of generality that $(x_1,v),(x_2,v) \in A_s$. (We will only need the fact that at least two of these edges are either both in $A_r$ or both in $A_s$. The case of at least two edges in $A_r$ has equivalent reasoning.)
	Both $(x_1,v)$ and $(x_2,v)$ can be paired with the edge it traverses $v$ with, i.e. its subsequent edge in the tour, as $(x_1,v),(x_2,v) \in A_s$ are positively oriented. Let $e_1, e_2$ be these subsequent edges. Then note that $C$ can be split into two smaller $(s-r)$ directed tours $C_1$ and $C_2$, with $C_1$ starting with edge $e_1$ and ending with $(x_2,v)$, and $C_2$ starting with edge $e_2$ and ending with $(x_1,v)$. This contradicts the assumption that $C$ was inclusion-wise minimal.  
	\end{subproof}

	We denote $T^+ = E(T)\setminus \supp(r)$ as the set of edges of $T$ that are not yet covered by $r$. 
	Hence, if $e\in T^+$, then $e \in A_s$ and there is at least one $C \in \mathcal{C}$ that contains $e$. 
	We choose for each $e \in T^+$ such an $(s-r)$ directed tour $C_e \in \mathcal{C}$ arbitrarily. 
	Let $\mathcal{C}^+=\{C_e:e \in T^+\}$ be the set of chosen $(s-r)$ directed tours. 
	We define $f$ as follows: for each $u,v \in V$ we set
	\begin{equation} \label{eq:f}
		f(u,v) = r(u,v) + (-1)^{[(u,v)\in A_r]} \sum_{C \in \mathcal{C}^+} [(u,v) \in C]. 
	\end{equation}
	In other words, $f$ is obtained from $r$ by removing one copy of edges in $C \cap A_r$ and adding one copy of edges in $C \cap A_s$ for all $C \in \mathcal{C}^+$. 

	Notice that $|\mathcal{C}^+| \le |T^+| \le |T|$.  By using~\eqref{eq:f} and subsequently~\eqref{eq:$(s-r)$ directedtour}, we get for all $v \in V$:
	\begin{align*}
		\sum_{u \in V} |r(u,v) -f(u,v)| &\le \sum_{u\in V}\sum_{C\in\mathcal{C}^+} [(u,v)\in C] \\
		&= \sum_{C\in\mathcal{C}^+}\sum_{u\in V} [(u,v)\in C]\\
		&\le \sum_{C\in\mathcal{C}^+} 2  \le 2|T|.
	\end{align*}
	Similarly we can conclude for all $v\in V$ that $\sum_{u \in V} |r(v,u) - f(v,u)| \le 2|T|$.

	\begin{numclaim} \label{claim:fisflow}
		For all $e\in T$, $f(e)>0$ and $f$ is a flow for the given instance.
	\end{numclaim}
		\begin{subproof}
		We first show that for all $e \in E$: 
		\[\min\{r(e),s(e)\} \le f(e) \le  \max\{r(e),s(e)\} .\]
		If $e \not \in A$, equation~\eqref{eq:f} implies that $r(e) = f(e) = s(e)$. 
		If $e \in A_s$, by definition $r(e)<s(e)$ and we can see from~\eqref{eq:f} that if the multiplicity of $e$ changes, it is because copies of $e$ are added to $r$ to form $f$ (and none are removed). 
		Because $m_A(e) \le s(e) - r(e)$, at most this many copies of $e$ can be added to $r$ to form $f$.  Hence $r(e) \le f(e) \le r(e) + (s(e)-r(e)) = s(e)$. Similarly, if $e \in A_r$, $r(e) > s(e)$ and at most $r(e) - s(e)$ copies of $e$ are removed from $r$ to form $f$. 
		
		Next we prove that $f$ is an allowed solution to the given \textsc{Flow} instance. 
		Let $C \in \mathcal{C}^+$ and let $e,e'$ be two subsequent edges from $C$ with common vertex $v$. If $e \in A_s$ and $e' \in A_r$, then the in- and out-degrees of $v$ do not change while adding a copy of $e$ and removing a copy of $e'$ (in equation~\eqref{eq:f}), because the orientation of $e$ and $e'$ is different. This is also true if $e\in A_r$ and $e' \in A_s$.
		If $e, e' \in A_r$, then both $e$ and $e'$ have a copy removed in equation~\eqref{eq:f}. Since the orientation of $e$ and $e'$ is the same, both the in- and out-degree of $v$ go down by one. We remark that this situation only happens if $v \not\in D$ by definition of $(s-r)$ directed tours. Similarly, if $e, e' \in A_s$, the in- and out-degree of $v$ increases by one. Since $r$ was an allowed solution, this implies that the number of incoming- and outgoing edges of $v$ in $f$ are equal, in other words, the flow is preserved. 
		Since for $v \in D$, the total incoming (and total outgoing) edges do not change, the demands are satisfied by $f$. Furthermore, the capacity constraints are satisfied since $ f(e) \le  \max\{r(e),s(e)\} \le \Ca(e)$.
		
		We show that $T\subseteq \supp(f)$. Let $e \in T^+$, then $e \in A_s$ so copies of $e$ are added to $r$ to form $f$ in equation~\eqref{eq:f}. Since $e \in T^+$, at least one tour $C \in \mathcal{C}^+$ contains $e$. Hence, $f(e)>0$. 
		If $e \in T \setminus T^+$, then $r(e)>0$ because $T^+ = T \setminus \supp(r)$ by definition. We also see that $s(e)>0$ because $T\subseteq \supp (s)$ by assumption. Using our earlier result that $\min\{r(e),s(e)\} \le f(e) $, we conclude that $f(e)>0$. 
	\end{subproof}
	
	We are left to prove that $\Dis(f)\le \Dis(s)$. For any $C \in \mathcal{C}$ define $\delta(C) = \Dis(A_s \cap C) - \Dis(A_r \cap C)$ as the cost of adding all edges in $A_s \cap C$ and removing all edges in $A_r \cap C$. Notice that $\delta(C)\ge 0 $ for all tours $C \in \mathcal{C}$, as otherwise $r$ would not have been optimal since we could improve it by augmenting along $C$. We note that $\sum_{C \in \mathcal{C}^+} \delta(C) \le \sum_{C \in \mathcal{C}}\delta(C)$ as $\mathcal{C}^+ \subseteq \mathcal{C}$. Therefore:
	$$\Dis(f) = \Dis(r) + \sum_{C \in \mathcal{C}^+} \delta(C) \le  \Dis(r) + \sum_{C \in \mathcal{C}} \delta(C) = \Dis(s).\vspace{-2.8em}$$
	\end{proof}

\subsection{Fixed Parameter Tractable algorithm}
\label{sec:dp}

Now we use Lemma~\ref{lem:s-rtours} to show that $\textsc{Connected Flow}$  is fixed-parameter tractable parameterized by the size of a vertex cover of $G$:

\begingroup
\def\thetheorem{\ref{thm:XFPT}}
\begin{theorem}
	There is an algorithm solving a given instance $(G,D,\Dem,\Dis,\Ca)$ of \textsc{Connected Flow} such that $G$ has a vertex cover of size $k$ in time $\Os(k^{\Oh(k)})$.
\end{theorem}
\addtocounter{theorem}{-1}
\endgroup

\begin{proof}
	Let $X$ be a vertex cover of size $k$ of $G=(V,E)$, let $s$ be an arbitrary optimal solution of \textsc{Connected Flow} and let $X'\subseteq X$ be the set of vertices of $X$ that are visited at least once by $s$. We will guess this set $X'$ as part of our algorithm, i.e. go through all possible sets. Hence we do the following algorithm for all $X'$ such that $(D \cap X) \subseteq X' \subseteq X$, which is at most $2^k$ times. 

	For any $X'$, we adjust $G$ such that the vertex cover is an independent set and all $x \in X'$ are visited at least once in any solution as follows. 
	We remove any edge $(x_i,x_j)\in E$ for $x_i, x_j \in X'$ and replace this edge by adding a new vertex $y$ to $V$. This $y$ has no demand and has edges $(x_i,y)$ and $(y,x_j)$, with capacities equal to the old capacity $\Ca(x_i,x_j)$ and $\Dis(x_i,y) = \Dis(x_i,x_j)$ and $\Dis(y,x_j) =0$. This removes any edges between vertices in the set $X'$, making it an independent set. We note that $X'$ is still a vertex cover of size $k$. 
	
	The vertices $x \in X'\cap D$ are visited at least once because of their demand. For all $x \in X'\setminus D$ we add a vertex $b_x$ to $V$, with $\Dem(b_x) =1$ and we add edges $(x,b_x)$ and $(b_x,x)$, both with $0$ cost and a capacity of $1$. As $b_x$ has a demand of $1$ and has only $x$ as its neighbor, this ensures that $x$ is visited at least once. 
	
	We remove all $x \in X \setminus X'$ from $V$ and denote the resulting graph as $G'=(V',E')$. Note that if $X'$ is guessed correctly, the optimal solution $s$ of the original instance, is also optimal for this newly created instance (by adding flow over the newly created edges between $x$ and $b_x$, and replacing any edge $(x_i,x_j)$ by the edges $(x_i,y)$ and $(y,x_j)$). We compute a \emph{relaxed} solution $r$ for this newly created instance, which can be done in polynomial time. 
	
	Let $T$ be any directed tree of size at most $2k$ such that $T \subseteq \supp (s)$ and all $x \in X'$ are incident to at least one edge $e \in T$. We argue why such tree $T$ exists. Since $s$ is connected and visits all $x\in X'$, we can find a tree $T\subseteq \supp(s)$ that covers all $x\in X'$. If $|T| > 2k$, we remove all the leaves from $T$ not in $X'$. Since $V'\setminus X'$ is an independent set (as $X'$ is a vertex cover), this means that the size of $T$ is bounded by $2k$. Note that all $x \in X'$ are still incident to an edge $e\in T$. 
	
	We apply Lemma~\ref{lem:s-rtours}, to $s$ and $T$ to find that there is a flow $f$ such that $\Dis(f) \le \Dis(s)$ and for every $v \in V$:
	\begin{equation}
		\label{eq:lemma}
		\sum_{u \in V'} |r(u,v) - f(u,v)| \le 4k, \qquad \text{ and } \qquad \sum_{u \in V'} |r(v,u) - f(v,u)| \le 4k.
	\end{equation}
	Since $T\subseteq \supp(f)$, $f$ visits all the vertices in $X'$ at least once. As $X'$ is a vertex cover, this means that $f$ is a connected flow and hence an optimal solution of the instance of \textsc{Connected Flow}.
	We will use a dynamic programming method to find solution $f$. Namely, we iteratively add vertices from the independent set $B=V'\setminus X'$ and keep track of the connectedness of our vertex cover $X'$ with a partition $\pi$. We later will restrict the number of table entries we actually compute with the help of equation~\eqref{eq:lemma}. 
	
	Denote $X' =\{x_1,\dots,x_{k'}\}$ and let $B=\{b_1,\dots,b_n\}$. 
	For $j \in [0,n]$ let $B_j$ be the set of the first $j$ vertices of $B$, i.e. $B_j=\{b_1,\dots, b_j \}$ and define $V_j = X' \cup B_j$. 
	For any $f:(V_j)^2 \to \N$ and $v \in V_j$ define \[f^{\Out}(v) = \sum_{u\in V_j}f(v,u) \,\,\text{ and }\,\, f^{\In}(v) = \sum_{u\in V_j}f(u,v).\] Let $\mathbf{c}^{\In} = (c_1^{\In},\dots,c_{k'}^{\In}) \in \N^{k'}$ and $\mathbf{c}^{\Out} = (c_1^{\Out},\dots,c_{k'}^{\Out}) \in \N^{k'}$ be two vectors of integers and let $\pi$ be a partition of the vertices of $X'$. 
	
	For $j\in [0,n]$ we define the dynamic programming table entry	$T_j(\pi,\mathbf{c}^{\In},\mathbf{c}^{\Out})$ to be equal to the minimal cost of any partial solution $f:(V_j)^2 \to \N$ having the specified in and out degrees ($\mathbf{c}^{\In}$ and $\mathbf{c}^{\Out}$) for vertices in $X'$ and connecting all vertices $x \in S$ for each $S \in \pi$.
	More formally, $T_j(\pi,\mathbf{c}^{\In},\mathbf{c}^{\Out})$ is equal to $\min_{f}\Dis(f)$ over all $f:(V_j)^2\to \N$ such that the following conditions hold:
	\begin{enumerate}\setlength\itemsep{0em}
	\item for all blocks $S$ of the partition $\pi$, the block is weakly connected in $G'_{f}$, 
	\item for all $x_i \in X'$: $f^{\Out}(x_i) = c_i^{\Out}$, $f^{\In}(x_i) = c_i^{\In}$,
	\item for all $v \in B_j$: $f^{\Out}(v) = f^{\In}(v)$, and if $v \in B_j\cap D$, then $f^{\In}(v) = \Dem(v)$,
	\item for all $u,v\in V_j: f(u,v) \le \Ca(u,v)$.
	\end{enumerate}
	We set $T_j(\pi,\mathbf{c}^{\In},\mathbf{c}^{\Out}) = \infty$ if no such $f$ exists. 
		
	\begin{numclaim}\label{claim:DP}
		Each table entry $T_j(\pi,\mathbf{c}^{\In},\mathbf{c}^{\Out})$ can be computed from all table entries $T_{j-1}$.
	\end{numclaim}
	\begin{subproof}
		Compute table entries for $j=0$ as follows. Set $T_0(\{\{x_1,\},\dots,\{x_{k'}\}\},\mathbf{0},\mathbf{0})$ to $0$, and all other entries of $T_0$ to $\infty$, as $V_0 = X'$ is an independent set and so a flow of zero on every edge is the only possible flow. 
		
		Now assume $j> 0$.
		We compute the values of $T_{j}(\pi,\mathbf{c}^{\In},\mathbf{c}^{\Out})$ as the minimum of the following value over all suitable $\mathbf{h}^{\In}_i = (h^{\In}_1,\dots,h^{\In}_{k'}) \in \N^{k'}$, $\mathbf{h}^{\Out}_i = (h^{\Out}_1,\dots,h^{\Out}_{k'}) \in \N^{k'}$, and all suitable partitions $\pi'$ of $X'$:
		\[T_{j-1}(\pi',\mathbf{c}^{\In}-\mathbf{h}^{\In},\mathbf{c}^{\Out}-\mathbf{h}^{\Out}) + \sum_{i=1}^{k'} \left( h^{\In}_i \cdot \Dis(b_{j},x_i) + h^{\Out}_i\cdot \Dis(x_i,b_{j})\right).
		\]
		Here we interpret $h^{\In}_i$ as the multiplicity of the edge $(b_j,x_i)$ and $h^{\Out}_i$ as the multiplicity of the edge $(x_i,b_j)$. Therefore, we require $h^{\In}_i \le \Ca(b_j,x_i)$ and $h^{\Out}_i \le \Ca(x_i,b_j)$ so that the capacity constraints hold. Furthermore, we require that the solution is flow preserving in $b_j$, i.e. $\sum_{i =1}^{k'} h^{\In}_i = \sum_{i =1}^{k'} h^{\Out}_i$ and $\sum_{i =1}^{k'} h^{\In}_i = \Dem(b_j)$ if $b_j \in D$. 
		For $\pi'$ we require for all $S\in \pi$ that either $S \in \pi'$ or there exist $S'_1, \dots S'_\ell \in \pi'$ such that $S_1'\cup \dots \cup S'_\ell = S$ and $S'_1,\dots,S'_\ell$ are all connected to $b_j$. This latter can be formalized by requiring that for each $t\in[1,\ell]$, there is an $x_i\in S'_t$ such that $h^{\In}_i + h^{\Out}_i >0$. 
		
		Notice that with this recurrence, the table entries are computed correctly as only the vertex $b_j$ was added, compared to the table entries $T_{j-1}$. Therefore we may assume that only the edges incident to $b_j$ were added to another solution for some table entry in $T_{j-1}$. 
	\end{subproof}
		
	We restrict this dynamic program using equation \eqref{eq:lemma}. As $X'$ is an independent set, there are only edges between $x\in X'$ and $b\in B_j$. Therefore, there exists a solution $f$ such that for every $x \in X'$ and $j \in [0,n]$:
	\begin{equation}
		\label{eq:lemmacor}
		\sum_{b \in B_j} |r(b,x) - f(b,x)| \le 4k, \qquad \text{ and } \qquad \sum_{b \in B_j} |r(x,b) - f(x,b)| \le 4k
	\end{equation}	
	
	 We restrict the dynamic program to only compute table entries $T_j$ respecting equation \eqref{eq:lemmacor}, by requiring for all $i\in[1,{k'}]$:
	\begin{equation}\label{eq:restrict}
		\begin{aligned}
			c_i^{\Out} &\in \left[\sum_{b \in B_j} r(x_i,b) - 4k, \sum_{b \in B_j} r(x_i,b) + 4k \right], \text{ and } \\
			c_i^{\In} &\in \left[\sum_{b \in B_j} r(b,x_i) - 4k, \sum_{b \in B_j} r(b,x_i) + 4k \right].
		\end{aligned}
	\end{equation}
	Note that the dynamic program is still correct with this added restriction, as $\sum_{b \in B_{j-1}} |r(b,x) - f(b,x)| \le \sum_{b \in B_{j}} |r(b,x) - f(b,x)| \le 4k$, so any table entry $T_j$ respecting equation~\eqref{eq:restrict} can be computed from all table entries $T_{j-1}$ respecting equation~\eqref{eq:restrict}.

	The dynamic program returns the minimum value of $T_n(\{X'\},\mathbf{c},\mathbf{c})$ for all $\mathbf{c}$ such that $c_i = \Dem(x_i)$ for all $x_i \in D\cap X'$. This returns the value of a minimum cost solution $f$ for $G'$, respecting equation~\eqref{eq:lemmacor}, if one exists. Let $f_{X'}$ be solution the dynamic program found in the iteration using $X'$. Then $\min\{f_{X'}: (D\cap X) \subseteq X' \subseteq X\}$ is equal to the minimum cost connected flow.  
	
	We count the number of different table entries $T_j$ computed by the dynamic program for fixed $j$. There are at most $(8k)^{k}$ possible values for both $\mathbf{c^{\In}}$ and $\mathbf{c^{\Out}}$ and at most $k^k$ different partitions $\pi$ of $X'$, so a total of $k^k\cdot(8k)^{2k}$ different entries. To compute one table entry of $T_j$, we only need (the at most $k^k\cdot (8k)^{2k}$) table entries of $T_{j-1}$. Note that we compute this dynamic programming table for each $X'$ such that $(D \cap X) \subseteq X' \subseteq X$, that is at most $2^k$ different $X'$. Hence the algorithm runs in time $\Os(k^{\Oh(k)})$.
\end{proof}

\subsection{Kernel for \textsc{Many Visits TSP} with $\Oh(k^5)$ vertices}
\label{sec:polyker}
We now present how to find a kernel with $\Oh(k^5)$ vertices for any instance of \textsc{MVTSP}, where $k$ is the size of a vertex cover of $G$. We do this by first finding an optimal solution $r$ to the relaxed \textsc{Flow} problem and then fixing the amount of flow on some edges based on this $r$. We prove that there is an optimal solution $s$ of \textsc{MVTSP} such that for all except $\Oh(k^5)$ vertices, all edges incident to these vertices have exactly the same flow in $r$ and $s$, as a consequence of Lemma~\ref{lem:s-rtours}. 

\begingroup
\def\thetheorem{\ref{thm:Xkernel}}
\begin{theorem}
	\textsc{MVTSP} admits a kernel polynomial in the size $k$ of the minimum vertex cover of $G$. 
\end{theorem}
\addtocounter{theorem}{-1}
\endgroup

\begin{proof}
	Fix an input instance on \textsc{MVTSP}. Let $k$ be the number of vertices in the vertex cover $X=\{x_1,\dots, x_k\}$ of $G$ and let $n$ be the size of the independent set $B = V\setminus X$. 
	Let $r$ be an optimal solution of the instance of \textsc{Flow} obtained by relaxing the connectivity constraint from in the given instance of \textsc{MVTSP}. 
	
	Define multisets $\overr{F}= (X \times B)\cap r$ (i.e. all edges in $r$ going from vertices in $X$ to vertices in $B$) and $\overl{F}= (B\times X)\cap r$. 
	
	\begin{numclaim}\label{claim:Fnocycles}
		We may assume that for both $\overr{F}$ and $\overl{F}$, their underlying undirected edge sets do not contain cycles.
	\end{numclaim}
	\begin{subproof}
		We change $r$ such that for both $\overr{F}$ and $\overl{F}$, their underlying undirected edge sets do not contain cycles. Assume that there is an \emph{alternating} cycle $C \subseteq \overl{F}$, meaning that its underlying edge set is a cycle and (hence) the edges alternate between being in positive and negative orientation. We can then create solutions $r'$ and $r''$ of \textsc{Flow} by alternatingly adding and removing edges from $C$. Note that we can start by either adding or removing, giving us these two different solutions $r'$ and $r''$. Since the edges added to $r$ to form $r'$ are exactly the edges that were removed from $r$ to form $r''$, and vice versa, it holds that $\Dis(r)-\Dis(r') = -(\Dis(r)-\Dis(r''))$. Since $r$ is an optimal solution, we conclude $\Dis(r)=\Dis(r')=\Dis(r'')$. We can therefore choose either $r'$ or $r''$ to replace $r$, such that $\overl{F}$ now has one alternating cycle less without changing any of the edges of $r$ outside $C$. Hence we can iteratively remove the cycles from $\overl{F}$ and $\overr{F}$ and obtain an optimal solution $r$ to the \textsc{Flow} instance in which both $\overl{F}$ and $\overr{F}$ are forests in polynomial time.
	\end{subproof}
	
	We partition $B$ as follows: $B=Y \cup \left(\bigcup_{i,j \in [1,k]} B_{ij} \right)$, where for each $b \in B_{ij}$: 
	$r(x_i,b) >0 $, $r(b,x_j) >0$, and
	\[r(x_a,b) = 0 \text{ for all } a \neq i \text{\qquad and \qquad} r(b,x_a) = 0 \text{ for all } a \neq j,\]
	and $Y = B \setminus \left(\bigcup_{i,j \in [1,k]} B_{ij} \right)$. 

	We argue that $|Y|\le k$. Recall that $m_{B}$ denotes the multiplicity function of a multiset $B$. Let $F = \supp(m_{\overl{F}}) \cup \supp(m_{\overr{F}})$ (note that $F$ is a set and not a multiset). Then $|F| \ge \sum_{i,j \in [1,k]} 2 |B_{ij}| + 3 |Y| = 2n + |Y|$, as any vertex $v \in B_{ij}$ must be responsible for exactly $2$ edges in $F$ and each vertex in $Y$ must add at least $3$ edges to $F$. Here we use that each vertex has a demand and therefore must have at least one incoming and outgoing edge from $r$. As $F$ is a union of two forests on $n+k$ vertices, we see that $|F| \le 2(n+k-1)$. We conclude that $2(n + |Y|) \le 2(n+k-1) $, i.e. $|Y| \le k$. 
	
	Let $s$ be an optimal solution of the \textsc{MVTSP} instance, so $s$ visits every vertex at least once. Hence there exists a directed tree $T\subseteq \supp(s)$, covering all vertices of $X$, of size at most $2k$. This tree exists by similar arguments as in the proof of Theorem~\ref{thm:XFPT}. We apply Lemma~\ref{lem:s-rtours} to $s$ and $T$, to find that there exists an optimal solution $f$ to the given \textsc{MVTSP} instance such that 
	\begin{equation}\label{eq:boundonf} \sum_{v\in V}\left(|r(x_i, v)-f(x_i, v)|+ |r(v,x_i)-f(v,x_i)| \right)\leq 8k\qquad \forall i\in [1,k].
	\end{equation}
	We note that $G_f$ is connected because $T \subseteq \supp(f)$ and $T$ connects all the vertices of the vertex cover.
	Equation~\eqref{eq:boundonf} implies that at most $8k^2$ edges of $\overl{F}$ and $\overr{F}$ are different in an optimal solution $f$ of \textsc{MVTSP} that is close compared to $r$.
		
	For every $i,j,\ell\in [1,k]$, we define $\overr{A_{ij}}(\ell)$ as the set of $8k^2+2$ vertices $v \in B_{ij}$ with the smallest values of $\Dis(x_\ell,v) -\Dis(x_i,v)$ (arbitrarily breaking ties if needed). 
	Intuitively, the vertices in $\overr{A_{ij}}(\ell)$ are the vertices for which re-routing the flow sent from $x_i$ to $v$ to go from $x_\ell$ to $v$ is the least expensive.
	Similarly we define $\overl{A_{ij}}(\ell)$ as a set of size $8k^2+2$ containing vertices $v \in B_{ij}$ with the smallest values of $\Dis(v,x_\ell) -\Dis(v,x_j)$.

	We also define a set $R_{ij}$ of `remainder vertices' as follows:
	$$R_{ij}=B_{ij}\setminus\left(\left(\bigcup_{\ell \in [1,k]} \overl{A_{ij}}(\ell)\right)\cup \left(\bigcup_{\ell \in [1,k]} \overr{A_{ij}}(\ell) \right)\right) \text{ for all  }i,j\in[1,k].$$ 
	
	\begin{numclaim}\label{claim:kernel}
		There exists an optimal solution $f'$ of the \textsc{MVTSP} instance such that for all $i,j \in [1,k]$, $b \in R_{ij}$ and $x_\ell \in X$ it holds that $r(x_\ell,b) = f'(x_\ell,b)$ and $r(b,x_\ell) = f'(b,x_\ell)$. 		
	\end{numclaim}

	\begin{subproof}
		We build this $f'$ iteratively from $f$, by removing any edges $(x_{i'},b)$ and $(b,x_{j'})$ for $i' \neq i$ and $j' \neq j$ for each $b \in R_{ij}$. In particular, this implies that $r(x_i,b) = f'(x_i,b)$ and $r(b,x_j) = f'(b,x_j)$, as $b$ then only has edges coming from $x_i$ and to $x_j$ and since $b$ has a fixed demand. 
		
		Throughout the process we retain optimality and connectivity for $f'$. Furthermore, after each step, the solutions $r$ and $f'$ differ at at most $8k^2$ edges. We start by setting $f' =f$.
		
		Let us consider $b \in R_{ij}$ and suppose that $f'(x_\ell,b) >0$ for some $\ell \not = i$. Note that we can tackle the case where $f'(b,x_\ell)>0$ for some $\ell \neq j$ with similar steps. We remark that $|\overr{A_{ij}}(\ell)| = 8k^2 +2$ as $R_{ij} \not = \emptyset$. As at most $8k^2$ edges are different between $r$ and $f'$, there are vertices $v,w \in \overr{A_{ij}}(\ell)$ such that all  of the edges adjacent to $v$ and $w$ have the same multiplicities in $r$ and $f'$, i.e. $f'(x,v) = r(x,v)$ and $f'(x,w) = r(x,w)$ for all $x\in X$. 
		
		Define flow $f''$ with at most the same costs as $f'$ by removing one copy of the edges $(x_\ell,b)$ and $(x_i,v)$ and adding one copy of the edges $(x_i,b)$ and $(x_\ell,v)$, see Figure~\ref{fig:kernel}. As $b \not \in \overr{A_{ij}}(\ell)$ and $v \in \overr{A_{ij}}(\ell)$, the cost of $f''$ is indeed at most the cost of $f'$ by definition of the set $\overr{A_{ij}}(\ell)$. 
		
		We now argue that $f''$ is connected. As $f'$ is a solution to \textsc{MVTSP}, it must be connected. Since we removed $(x_\ell,b)$ and $(x_i,v)$ from $f'$ to form $f''$, proving that the pairs $x_\ell,b$ and $x_i,v$ are connected in $f''$ proves $f''$ to be connected. The edges $(x_i,w)$, $(w,x_j)$ and $(v,x_j)$ in $f''$ connect $x_i$ and $v$. As a consequence, $x_\ell$ and $b$ are also connected, because of the edges $(x_\ell,v)$ and $(x_i,b)$.
		
		We remark that the number of edges that differ between $f''$ and $r$ has not changed. Hence, we continue with setting $f' = f''$ and repeat until $f'$ has the required properties.
		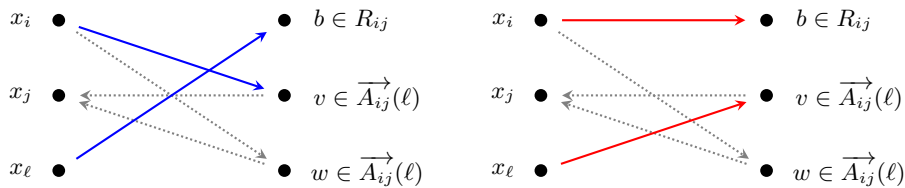
\begin{figure}[ht!]
			\centering
			\begin{minipage}{.45\textwidth}
				\centering
				\scalebox{1}{\begin{tikzpicture}[stealth-stealth, shorten >=8pt, shorten <=8pt,thick ]

	%\node[] at (0,9.5) {\Large $X$};
	
	%\draw[blue] (0,0) to[out = -0,in = -140] (2,0.5);
	%\draw[blue] (2,0.5) to[out = 130,in = -0] (0,1);
	%\draw[blue,rounded corners= 30pt]  (0,0)-- (2.6,0.5) --  (0,1);

	%\draw[blue] (0,4) to[ out = 0, in = -145] (2,4);
	%\draw[blue] (0,5) to[ out = 0, in =145] (2,5);
	%\draw[blue] (2,5) to[out = 0, in = 0] (2,4);
	
	\filldraw[black] (0,0) circle (2pt);
	\filldraw[black] (0,1) circle (2pt);
	\filldraw[black] (0,2) circle (2pt);
	
	\filldraw[black] (3,0) circle (2pt);
	\filldraw[black] (3,1) circle (2pt);
	\filldraw[black] (3,2) circle (2pt);
	
	\node[] at (-0.5,2) {$x_i$};
	\node[] at (-0.5,1) {$x_j$};
	\node[] at (-0.5,0) {$x_\ell$};

	\node[] at (3.9,2) {$b \in R_{ij}$};
	\node[] at (4.1,1) {$v \in \overr{A_{ij}}(\ell)$};
	\node[] at (4.1,0) {$w \in \overr{A_{ij}}(\ell)$};
	
	\draw[gray, densely dotted, -stealth] (0,2) -> (3,0);
	\draw[gray, densely dotted, -stealth] (3,0) -> (0,1);
	\draw[gray, densely dotted, -stealth] (3,1) -> (0,1);
	
	\draw[blue, -stealth] (0,0) -> (3,2);
	\draw[blue, -stealth] (0,2) -> (3,1);

\end{tikzpicture}}
			\end{minipage} \hspace{20pt}
			\begin{minipage}{.45\textwidth}
				\centering
				\scalebox{1}{\begin{tikzpicture}[stealth-stealth, shorten >=8pt, shorten <=8pt,thick ]

	%\node[] at (0,9.5) {\Large $X$};
	
	%\draw[blue] (0,0) to[out = -0,in = -140] (2,0.5);
	%\draw[blue] (2,0.5) to[out = 130,in = -0] (0,1);
	%\draw[blue,rounded corners= 30pt]  (0,0)-- (2.6,0.5) --  (0,1);

	%\draw[blue] (0,4) to[ out = 0, in = -145] (2,4);
	%\draw[blue] (0,5) to[ out = 0, in =145] (2,5);
	%\draw[blue] (2,5) to[out = 0, in = 0] (2,4);
	
	\filldraw[black] (0,0) circle (2pt);
	\filldraw[black] (0,1) circle (2pt);
	\filldraw[black] (0,2) circle (2pt);
	
	\filldraw[black] (3,0) circle (2pt);
	\filldraw[black] (3,1) circle (2pt);
	\filldraw[black] (3,2) circle (2pt);
	
	\node[] at (-0.5,2) {$x_i$};
	\node[] at (-0.5,1) {$x_j$};
	\node[] at (-0.5,0) {$x_\ell$};

	\node[] at (3.9,2) {$b \in R_{ij}$};
	\node[] at (4.1,1) {$v \in \overr{A_{ij}}(\ell)$};
	\node[] at (4.1,0) {$w \in \overr{A_{ij}}(\ell)$};
	
	\draw[gray, densely dotted, -stealth] (0,2) -> (3,0);
	\draw[gray, densely dotted, -stealth] (3,0) -> (0,1);
	\draw[gray, densely dotted, -stealth] (3,1) -> (0,1);
	
	\draw[red, -stealth] (0,0) -> (3,1);
	\draw[red, -stealth] (0,2) -> (3,2);

\end{tikzpicture}}
			\end{minipage}
			\caption{Adjusting flow $f'$, depicted on the left, to get flow $f''$, depicted on the right. The blue edges are replaced by the red edges, the rest of the solutions are equal. The vertex $w$ assures the new solution remains connected} 
			\label{fig:kernel}
		\end{figure}
	\end{subproof}
	
	Therefore, we may assume that, in $G_{f'}$, the vertices in $R_{ij}$ are adjacent only to $x_i$ and $x_j$ for all $i,j\in[1,k]$. This proves that the following reduction rule is correct: contract all vertices in $R_{ij}$ into one vertex $r_{ij}$ with edges only $(x_i,r_{ij})$ and $(r_{ij},x_j)$ of cost zero and let the demand $\Dem(r_{ij}) =\sum_{v \in R_{ij}} \Dem(v)$. Hence, we require any solution to use the vertices in $r_{ij}$ exactly the number of times that we would traverse all the vertices of $R_{ij}$. 
	By applying this rule, we get a kernel with the vertices from the sets $X$, $Y$, $\overl{A_{ij}}(\ell)$, $\overr{A_{ij}}(\ell)$, and $r_{ij}$, which is of size $$|X|+|Y|+\sum_{i,j,\ell \in [1,k]} \left( \left|\overl{A_{ij}}(\ell)\right| +  \left|\overr{A_{ij}}(\ell)\right| \right)+ k^2\leq k+k+k^3\cdot (8k^2+2) + k^2=\Oh(k^5).$$
	To subsequently reduce all costs to be at most $2^{k^{\Oh(1)}}$ we can use a method from Etscheid et al.~\cite{etscheid2017polynomial} in a standard manner.

	We check that one can construct this kernel in polynomial time. First, compute a relaxed solution $r$ and remove any cycles in $\overr{F}$ and $\overl{F}$ in polynomial time. Next for each $i,j,\ell \in [1,k]$, compute in polynomial time what the sets $\overr{A_{ij}}(\ell)$ and $\overl{A_{ij}}(\ell)$ should be, by computing the values of $\Dis(x_\ell,v) -\Dis(x_i,v)$ and sorting. Finally we can contract all vertices in $R_{ij}$ into a vertex $r_{ij}$ polynomial time for all $i,j \in [1,k]$. 
\end{proof}
\section{Parameterisation by Treewidth} \label{sec:pathw}
In this section we consider the complexity of \textsc{Connected Flow}, when parameterized by the treewidth $\tw$ of $G$. We first give a $|V(G)|^{\Oh(\tw)}$ time dynamic programming algorithm for \textsc{Connected Flow}. Subsequently, we give a matching conditional lower bound on the complexity of \textsc{MVTSP} parameterized by the pathwidth of $G$. Since \textsc{MVTSP} is a special case of \textsc{Connected Flow} this shows that our dynamic programming algorithm is in some sense optimal.

\subsection{An XP algorithm for \textsc{Connected Flow}}
In this subsection we show the following:

\begingroup
\def\thetheorem{\ref{thm:TwDP}}
\begin{theorem}
Let $M$ be an upper bound on the demands in the input graph $G$, and suppose a tree decomposition of width $\tw$ of $G$ is given. Then a \textsc{Connected Flow} instance on $G$ can be solved in time $|V(G)|^{\Oh(\tw)}$ and an \textsc{MVTSP} instance on $G$ can be solved in time $\min\{|V(G)|,M\}^{\Oh(\tw)}|V(G)|^{\Oh(1)}$.
\end{theorem}
\addtocounter{theorem}{-1}
\endgroup

\begin{proof}
The algorithm is based on a standard dynamic programming approach; we only describe the table entries and omit the recurrence to compute table entries since it is standard.
We assume we have a tree decomposition $\mathcal{T} = (\{X_i\}, R )$ on the given graph. For a given bag $X_i$, let $\pi$ be a partition on $X_i$. Furthermore let $\mathbf{d}^{\In} = (d_v^{\In})_{v \in X_i} \in \N^{X_i}$ and $\mathbf{d}^{\Out} = (d_v^{\Out})_{v \in X_i} \in \N^{X_i}$ be two vectors of integers, indexed by $X_i$. We define the dynamic programming table entry $T(X_i, \pi, \mathbf{d}^{\In}, \mathbf{d}^{\Out})$ to be the cost of the cheapest partial solution on the graph `below' the bag $X_i$, among solutions whose connected components agree with the partition $\pi$ and whose in and out degrees agree with the vectors $\mathbf{d}^{\In}$ and $\mathbf{d}^{\Out}$. More formally, for $r \in V(R)$ the root of the tree decomposition, we consider a bag $X_j$ to be below another bag $X_i$ if one can reach $j$ from $i$ by a directed path in the directed tree obtained from $R$ by orienting every edge away from $r$. We will denote this as $X_j \preccurlyeq X_i$ and define $Y_i = \cup_{X_j \preccurlyeq X_i}X_j$. For each bag $X_i$, a partition $\pi$ of $X_i$ and sequences $\mathbf{d}^\In$ and $\mathbf{d}^\Out$ satisfying $(i)$ $0\leq d^\In_v,d^\Out_v \leq \Dem(v)$ for each $v \in D$ and $(ii)$ $0\leq d^\In(v),d^\Out(v) \leq M |V(G)|$ for each $v \notin D$, define $T(X_i, \pi, \mathbf{d}^{\In}, \mathbf{d}^{\Out}) = \min_s \Dis(s)$ over all $s : Y_i^2  \to \N$ such that the following conditions hold:

\begin{enumerate}\setlength\itemsep{0em}
\item $\sum_{u \in Y_i} s(u,v) = \sum_{u \in Y_i} s(v,u) = \Dem(v)$ for all $v \in D \cap (Y_i \setminus X_i)$,
\item $\sum_{u \in Y_i} s(u,v) = d^{\In}_v$ for all $v \in X_i$,
\item $\sum_{u \in Y_i} s(v,u) = d^{\Out}_v$ for all $v \in X_i$,
\item all blocks of the partition $\pi$ are weakly connected in $G_s$,
\item $s(u,v) \leq \Ca(u,v)$ for all $(u,v) \in E(G[Y_i])$.
\end{enumerate}
We can compute the table starting at the leaves of $R$ and work our way towards the root.

Let us examine the necessary size of this dynamic programming table. First we note that there are at most $|V(G)|^{\Oh(1)}$ bags in the tree decomposition. Next we consider the values $d_v^{\In}$ and $d_v^{\Out}$. Note that we can assume that an optimal solution only visits any vertex without demand at most $M|V(G)|$ times: Any solution can be decomposed into a collection of paths between vertices with demand. Each such path can be assumed to not visit any vertex more than once (except possibly in the end points of the path) since the solution is of minimum weight and all costs are non-negative. We find that each vertex gets visited at most $M|V(G)|$ times and thus we only need to consider $M|V(G)|$ many values of $d_v^{\In}$ and $d_v^{\Out}$. Thus the degree values of the partial solutions contribute a factor of $(M|V(G)|)^{\Oh(\tw)}$ to the overall running time of the algorithm if the given instance is a \textsc{Connected Flow} instance, and only $M^{\Oh(\tw)}$ if the given instance is an \textsc{MVTSP} instance (in which all vertices are demand vertices). 

We argue that we may assume that $M=|V(G)|^{\Oh(1)}$. Together with the fact that the number of possibilities for $\pi$ is $\tw^{\Oh(\tw)}\le |V(G)|^{\Oh(\tw)}$, the claimed result for \textsc{Connected Flow} follows. We support this assumption using a variation on the proof of Theorem 3.4 from Kowalik et al.~\cite{kowalik_et_al:LIPIcs:2020:12932}. Let $r$ be some optimal solution to \textsc{Flow}, then by applying Lemma~\ref{lem:s-rtours} with $T$ being some subtree of $G_s$ spanning all demand vertices, we find that there is some optimal solution $s$ of \textsc{Connected Flow} such that $|r(u,v) - s(u,v)| \leq 2n$ at every edge $(u,v)$. 

We now construct a flow $f$ from $r$ by subtracting simple directed cycles from $r$. Note that each time that we subtract such a cycle, the result is again a flow. We start with $f = r$ everywhere. Now if there is an edge $(u,v) \in E$ for which $f(u,v) > \max\{ r(u,v) - 2n - 1, 0 \}$, we can find a simple directed cycle $C \in G_f$, containing $(u,v)$, as $f$ is a flow and thus $G_f$ is Eulerian. Then define $f'(u,v) = f(u,v) - [(u,v) \in C]$. Note that $f'$ is again a flow. Set $f = f'$. We then repeat this process of subtracting simple directed cycles from $f$ until $f(u,v) \leq \max\{ r(u,v) - 2n - 1, 0 \}$ for every edge $(u,v)$. 

Note that $0\le s(u,v)-f(u,v)$ for all $(u,v)\in E$. 
Then define the instance with  $\Dem'(v) = \Dem(v) - \sum_{u \in V} f(u,v)$ and $\Ca'(u,v) = \Ca(u,v) - f(u,v)$ for which $s(u,v) - f(u,v)$ is an optimal connected flow. If $\Dem'(v) \leq 2n^2 + n$ we are done. Otherwise let $r'$ be a relaxed solution for the new instance. Note that there is some edge $(u',v')$ for which $r'(u',v') > 2n + 1$ and thus we can repeat the previous argument to find a non-zero flow $f'$ such that $f'(u,v) \leq \max\{ r'(u,v) - 2n - 1, 0 \}$ on every edge and define a corresponding new instance. Since each time we subtract a non-zero flow, after some number of repetitions we find $\Dem'(v) \leq 2n^2 + n$. 

For the result for \textsc{MVTSP}, the above approach would give a running time of $\min\{|V(G)|, M\}^{\Oh(\tw)}\tw^{\Oh(\tw)}|V(G)|^{O(1)}$. However, the factor $\tw^{\Oh(\tw)}$ in the running time needed to keep track of all partitions $\pi$ can be reduced to $2^{\Oh(\tw)}$ via a standard application of the rank based approach (see e.g.~\cite[Section 11.2.2]{DBLP:books/sp/CyganFKLMPPS15} or \cite{BodlaenderCKN15,DBLP:journals/corr/abs-1211-1506}).
\end{proof}

\subsection{Lower bound}\label{sec:lowerb}
We now present a modified version of a reduction from \textsc{3-CNF-SAT} to \textsc{Hamiltonian Cycle} parameterized by pathwidth from Cygan et al.~\cite{DBLP:journals/corr/abs-1211-1506}. We modify it to be a reduction to \textsc{MVTSP} instead.

We will produce an instance of \textsc{MVTSP} that is symmetric in the sense that the graph $G$ is undirected, hence we denote edges as unordered pairs of vertices (i.e. $\{u,v\} = \{v,u\}$). As a consequence, when $c$ is a tour on $G$, then we say $c(u,v) = c(v,u)$.
The general proof strategy is as follows. For a given \textsc{3-CNF-SAT} formula $\phi$ on $n$ variables\footnote{In this section, we will only use $n$ to refer to the number of variables of a \textsc{3-CNF-SAT} instance.} we will construct an equivalent \textsc{MVTSP} instance $(G,d)$. This graph will consist of $n/s$ paths, for some value $s$, with each path propagating some information encoding the value of $s$ variables of $\phi$. For each clause of $\phi$ we will add a gadget which checks if the assignment satisfies the clause. We then bound the size and the pathwidth of the constructed graph $G$. This allows us to conclude a lower bound based on this reduction.

\subsubsection{Gadgets}\label{subsec:3SATGad}
We start by borrowing the following gadget from Cygan et al.\cite{DBLP:journals/corr/abs-1211-1506}, called a 2-label gadget.

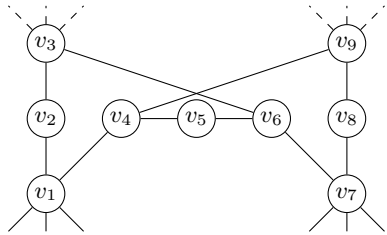
\begin{figure}[H]
    \centering
    \begin{tikzpicture}[-]
    \tikzstyle{every state} = [inner sep = 0.1mm, minimum size = 5mm]
    \node[state]		(V1) at (0,0)		{$v_1$};
    \node[state]		(V2) at (0,1)		{$v_2$};
    \node[state]		(V3) at (0,2)		{$v_3$};
    \node[state]		(V4) at (1,1)		{$v_4$};
    \node[state]		(V5) at (2,1)		{$v_5$};
    \node[state]		(V6) at (3,1)		{$v_6$};
    \node[state]		(V7) at (4,0)		{$v_7$};
    \node[state]		(V8) at (4,1)		{$v_8$};
    \node[state]		(V9) at (4,2)		{$v_9$};
    
    \path (V1)		edge			node	{}	(V2)
    					edge			node	{}	(V4)
    					edge			node	{}	(-0.5,-0.5)
    					edge			node	{}	(0,-0.5)
    					edge			node	{}	(0.5,-0.5)
    		(V2)		edge			node	{}	(V3)
    		(V3)		edge			node	{}	(V6)
    					edge[dashed]	node	{}	(-0.5,2.5)
    					edge[dashed]	node	{}	(0,2.5)
    					edge[dashed]	node	{}	(0.5,2.5)
    		(V4)		edge			node	{}	(V5)
    					edge			node	{}	(V9)
    		(V5)		edge			node	{}	(V6)
    		(V6)		edge			node	{}	(V7)
    		(V7)		edge			node	{}	(V8)
    					edge			node	{}	(3.5,-0.5)
    					edge			node	{}	(4,-0.5)
    					edge			node	{}	(4.5,-0.5)
		(V8)		edge			node	{}	(V9)
		(V9)		edge[dashed]	node	{}	(3.5,2.5)
    					edge[dashed]	node	{}	(4,2.5)
    					edge[dashed]	node	{}	(4.5,2.5);
    \end{tikzpicture}
    \caption{A 2-label gadget.}
    \label{fig:2LabGad}
\end{figure}

The key feature of this gadget is that if all vertices in the gadget have demand equal to 1, then if a solution tour enters the gadget at $v_3$, it has to leave the gadget at $v_9$ and vice versa. A similar relation holds for $v_1$ and $v_7$. We will refer to any edge connected to either $v_1$ or $v_7$ as having label 1 and any edge connected to $v_3$ or $v_9$ as having label 2. We will use this gadget to construct a gadget that can detect certain multisets of edges in a part of a graph. In this construction we will chain 2-label gadgets together using label 1 edges. Whenever we do this, we always connect the vertex $v_7$ of one gadget to the vertex $v_1$ in the next. To keep things concise, in the rest of this section we will refer to any 2-label gadget as if it were a single vertex.

This next gadget is also inspired by a construction from Cygan et al.~\cite{DBLP:journals/corr/abs-1211-1506}.

\begin{definition}
A scanner gadget in an unweighted \textsc{MVTSP} instance $(G, d)$ is described by a tuple $(X,a,b,\mathcal{F})$, where $X\subseteq V$, $a,b \in V \setminus X$ with $\Dem(a) = \Dem(b) = 1$, $\mathcal{F}$ is a family of multisets of edges in\footnote{Here $E(X,X)$ are all edges with both endpoints in $X$ and the restriction $c_Y$ are all edges in $c$ in $Y$ (keeping multiplicities).} $E(X,X)$ and $\emptyset \notin \mathcal{F}$. A tour $c$ of $G$ is \emph{consistent with $(X,a,b,\mathcal{F})$} if its restriction $c_{E(X,X)}$ is in $\mathcal{F}$ and if $c(a,b) > 0$.
\end{definition}

When refering to the gadget as a subgraph, we will use $G_\mathcal{F}$. We implement the scanner gadget using the following construction, obtaining a different instance $(G', \Dem')$ of MVTSP.

\begin{itemize}
    \item Remove the edges in $E(X,X)$.
    \item Add an independent set $I = \{ s_1, \dots, s_\ell \}$ and edges $\{ a, s_1 \}$ and $\{ s_\ell, b \}$, for $\ell = |\mathcal{F}|$.
    \item Let $\mathcal{F} = \{ F_1, \dots, F_\ell \}$. For $i = 1, \dots, \ell$ we do the following.
        \begin{itemize}
            \item Let $F_i = \{e_1^{q_1}, \dots, e_z^{q_z}\}$, that is $F_i$ contains $q_i$ copies of $e_i$.
            \item Add a path $P_i = \{p_i^{1}, \dots, p_i^{t_i}\}$ of 2-label gadgets, where $t_i = |F_i| = \sum_{i=j}^z q_j$. We connect the gadgets in a chain using label 1 edges.
            \item Connect $p_i^1$ to $s_{i-1}$ and $s_i$ using label 1 edges (green edges in Figure~\ref{fig:3SATGad}) and connect $p_i^{|F_i|}$ to $s_i$ and $s_{i+1}$ using label 1 edges (blue edges in Figure~\ref{fig:3SATGad}). 
            \item For all $j = 1, \dots, z$ add label 2 edges from $x$ $p_i^{j'}$ and from $y$ to $p_i^{j'}$ for $e_j = \{ x, y \}$ and for $q_j$ different, previously unused values of $j'$ (red edges in Figure~\ref{fig:3SATGad}).
        \end{itemize}
    \item We set the demand of all added vertices to $1$.
\end{itemize}

\begin{figure}[H]
    \centering
    \begin{tikzpicture}[-,scale=0.9]
    \def \GadWid{13}
    \tikzstyle{every state} = [inner sep = 0.1mm, minimum size = 6mm]
    \node[state]				(A) at (0,5)				{$a$};
    \node[state]				(B) at (\GadWid,5)		{$b$};
    \node						(L1) at (0,0)				{};
    \node						(R1) at (\GadWid,0)		{};
    \node						(X) at (\GadWid/2+1,1)	{$X$};
    \node	[state]				(X1) at (1,1)				{$x_1$};
    \node	[state]				(X2) at (2.5,1)				{$x_2$};
    \node	[state]				(X3) at (3.5,0.5)			{$x_3$};
    \node	[state]				(X4) at (5,1)			{$x_4$};
    \node	[state]				(X5) at (6,1.5)			{$x_5$};
    \node	[state]				(X6) at (\GadWid-3,1)	{$x_6$};
    \node	[state]				(X7) at (\GadWid-1,1)	{$x_7$};
    
    \node[state]				(S1) at (1.5,5)			{$s_1$};
    \node[state]				(S2) at (4.5,5)			{$s_2$};
    \node[state, draw=none]	(S3) at (7.5,5)			{};
    \node[state, draw=none]	(S-) at (\GadWid-4.5,5)	{};
    \node[state]				(Sl) at (\GadWid-1.5,5)	{$s_\ell$};
    \node[state]				(P11) at (1,3)			{$p_1^1$};
    \node[state]				(P12) at (2.5,3)			{$p_1^{t_1}$};
    \node[state]				(P21) at (4,3)			{$p_2^1$};
    \node[state]				(P22) at (5.5,3)			{$p_2^{t_2}$};
    \node[state, draw=none]	(P31) at (7,3)			{};
    \node[state, draw=none]	(P-2) at (\GadWid-4,3)	{};
    \node[state]				(Pl1) at (\GadWid-2.5,3)	{$p_\ell^1$};
    \node[state]				(Pl2) at (\GadWid-1,3)	{$p_\ell^{t_\ell}$};
    
    \node						at ($(S3)!0.5!(S-)$)	{$\dots$};
    \node						at ($(P31)!0.5!(P-2)$)	{$\dots$};
    
    \path[line width = 0.75pt]
    		(A)		edge							node	{}	(S1)
    		(B)		edge							node	{}	(Sl)
    		(S1)	edge[draw = green]				node	{}	(P11)
    				edge[draw = green]				node	{}	(P21)
    		(S2)	edge[draw = green]				node	{}	(P21)
    				edge[draw = green]				node	{}	(P31)
    				edge[draw = blue]				node	{}	(P12)
    				edge[draw = blue]				node	{}	(P22)
		(S3)	edge[draw = blue]				node	{}	(P22)
		(S-)		edge[draw = green]				node	{}	(Pl1)
    		(Sl)		edge[draw = blue]				node	{}	(P-2)
    				edge[draw = blue]				node	{}	(Pl2)
		(P11)	edge[draw = red]				node	{}	(X1)
				edge[draw = red]				node	{}	(X2)
				edge[dashed, out = 40, in = 140]	node	{}	(P12)
		(P12)	edge[draw = red]				node	{}	(X2)
				edge[draw = red]				node	{}	(X3)
		(P21)	edge[draw = red]				node	{}	(X3)
				edge[draw = red]				node	{}	(X4)
				edge[dashed, out = 40, in = 140]	node	{}	(P22)
		(P22)	edge[draw = red]				node	{}	(X3)
				edge[draw = red]				node	{}	(X5)
		(Pl1)	edge[draw = red]				node	{}	(X6)
				edge[draw = red]				node	{}	(X7)
				edge[dashed, out = 40, in = 140]	node	{}	(Pl2)
		(Pl2)	edge[draw = red]				node	{}	(X6)
				edge[draw = red]				node	{}	(X7);
				
	\draw (0,0) rectangle (\GadWid, 2);
    \end{tikzpicture}
    \caption{Example of the scanner gadget.}
    \label{fig:3SATGad}
\end{figure}
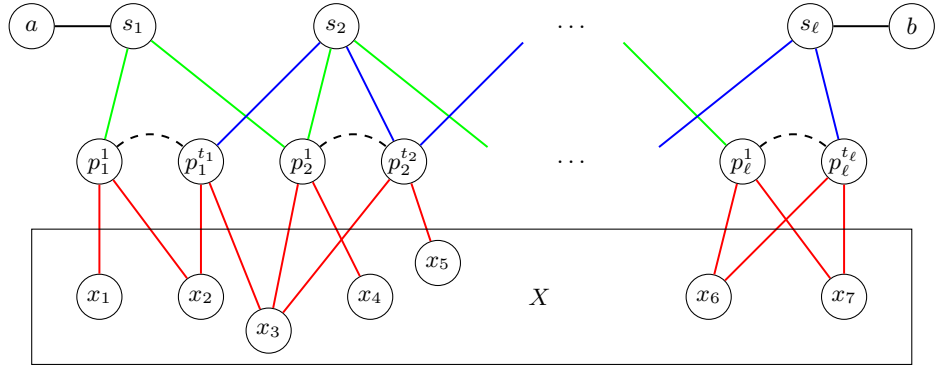

The function of the gadget is captured by the following lemma.

\begin{lemma} \label{lem:ConGad}
There exists an tour on $(G, \Dem)$ that is consistent with $(X,a,b,\mathcal{F})$ if and only if there exists a tour on $(G', \Dem')$.
\end{lemma}

The proof will closely follow that in Cygan et al.~\cite{DBLP:journals/corr/abs-1211-1506}.

\begin{proof}
Suppose we have a tour on $(G,\Dem)$ which is consistent with a gadget $(X, a, b, \mathcal{F})$. Let $F_i \in \mathcal{F}$ be the restriction of the tour on $E(X,X)$. Then the tour on $(G, \Dem)$ can be extended to a tour on $(G', \Dem')$ by replacing the $q_j$ instances of an edge $\{u,v\} \in F_i$ with two edges $\{u,p_i^{j'}\}$ and $\{v,p_i^{j'}\}$ for $q_j$ different values of $j'$. We also replace the edge $\{a,b\}$ by the path
\[
a, s_1, P_1, \dots,P_{i-1}, s_i, P_{i+1}, s_{i+1}, \dots, P_\ell, s_\ell, b.
\]
Since the obtained tour visits all vertices in the gadget exactly once and since the restriction of the adjusted tour connects the same pairs of vertices in $X$ as the restriction of the original tour, the obtained tour will be a solution for the instance $(G',\Dem')$.

For the other direction, suppose we have a tour $c'$ on $(G', \Dem')$. Note that by the nature of the 2-label gadgets any tour cannot cross from some $s_i$ into $X$ through one of the 2-label gadgets in one of the paths $P_i$. Thus the tour can only travel from outside the gadget to $s_i$, by going through $a$ or $b$.
Therefore the tour must include the edges $\{a,s_1\}$ and $\{s_\ell,b\}$. Furthermore $s_1$ and $s_\ell$ must be connected by some path $P'$ in the tour. Because $I$ is an independent set, $P'$ has to jump back and forth between the $P_i$'s and the $s_i$'s and has to include every $s_i$, since this is the only way to reach a vertex $s_i$ with a tour. 

This means that there will be exactly one path $P_{i_0}$ which is not covered by $P'$. We can now obtain a tour $c$ of $(G, \Dem)$ by first setting $c(u,v) = c'(u,v)$ for $\{u,v\} \neq \{a,b\}$ for $u$ or $v$ not in $X$. We then include any edge in $X$ a number of times according to its multiplicity in $F_{i_0}$ i.e. we set $c(u,v) = F_{i_0}(u,v)$. Finally we set $c(a,b) = 1$. Note that since $c(a,b) >0$ and $F_{i_0} \in \mathcal{F}$, we find that $c$ is consistent with $(X,a,b,\mathcal{F})$.
\end{proof}

The following lemma will allow us to implement the gadget without increasing the pathwidth of the graph too much.

\begin{lemma} \label{lem:GadPW}
The scanner gadget has pathwidth at most $|X| + 21$.
\end{lemma}

\begin{proof}
We define the bags of the decomposition as follows
\begin{align*}
B_a &:= X \cup \{a, s_1\} \\
B_{i,j} &:= X \cup \{s_{i-1}, s_i, s_{i+1}, p_i^j, p_i^{j+1}\}\\
B_b &:= X \cup \{b, s_\ell\}.
\end{align*}
We now have the following path decomposition of $G_\mathcal{F}$
\[
B_a,B_{1,1},B_{1,2},\dots B_{1,t_1},B_{2,1}, \dots B_{\ell,t_\ell},B_b.
\]
It is easy to check that every vertex/edge is covered by some bag and that for every vertex $v$ the set of bags containing $v$ form an interval in the decomposition. 
\end{proof}

\subsubsection{Construction}
Suppose we are given a \textsc{3-CNF-SAT} formula ${\phi = C_1 \wedge \dots \wedge C_m}$. We will construct an equivalent unweighted \textsc{MVTSP} instance $\Gamma_\phi$ using scanner gadgets. We will interpret a tuple $(q, j) \in \{1, \dots, 2^s\} \times \{1, \dots, n/s\}$ as an assignment of $x_{(j-1)s + 1}, \dots, x_{js}$ by first decomposing
\[
q - 1 = \sum \limits _{i = 1} ^s c_i 2^{i-1}
\]
and setting $x_{(j-1)s + i}$ as true if $c_i = 1$ and false if $c_i = 0$. We say a clause $C$ is satisfied by a set $Q$ of such tuples, if $j \neq j'$ for all $(q,j), (q,j') \in Q$, and if the partial assignment collectively given by the tuples satisfies $C$.

\begin{figure}[h]
    \centering
    \def \ConHig{6}
    \def \ConWid{10}
    \begin{tikzpicture}[-,scale=0.9]
    \tikzstyle{gadget} = [draw=black, very thick, rectangle, rounded corners, inner sep=4pt]
    \tikzstyle{every state} = [inner sep = 0.5mm, minimum size = 8mm]
    
    \path[use as bounding box] (-2,-1) rectangle (\ConWid+1.5, \ConHig+2);
    
    \node[state]		(L11) at (0,\ConHig - 2.5)				{$l_{1,1}$};
    \node[state]		(L12) at (0,\ConHig - 4)				{$l_{1,2}$};
    \node[state]		(L1n) at (0,0)						{$l_{1,\frac{n}{s}}$};
    \node[state]		(R11) at (2,\ConHig - 2.5)			{$r_{1,1}$};
    \node[state]		(R12) at (2,\ConHig - 4)				{$r_{1,2}$};
    \node[state]		(R1n) at (2,0)						{$r_{1,\frac{n}{s}}$};
    \node[state]		(Lm1) at (\ConWid - 2,\ConHig - 2.5)	{$l_{m,1}$};
    \node[state]		(Lm2) at (\ConWid - 2,\ConHig - 4)	{$l_{m,2}$};
    \node[state]		(Lmn) at (\ConWid - 2,0)				{$l_{m,\frac{n}{s}}$};
    \node[state]		(Rm1) at (\ConWid,\ConHig - 2.5)		{$r_{m,1}$};
    \node[state]		(Rm2) at (\ConWid,\ConHig - 4)		{$r_{m,2}$};
    \node[state]		(Rmn) at (\ConWid,0)				{$r_{m,\frac{n}{s}}$};
    \node[gadget]	(G1) at (1,\ConHig)					{$G_{C_1}$};
    \node[gadget]	(G2) at (4,\ConHig)					{$G_{C_2}$};
    \node[gadget]	(Gm) at (\ConWid - 1,\ConHig)		{$G_{C_m}$};
    
    \node[state, minimum size = 5mm]		(A1) at (0,\ConHig-1)			{$a_1$};
    \node[state, minimum size = 5mm]		(A2) at (2.5,\ConHig)			{$a_2$};
    \node[state, minimum size = 5mm]		(A3) at (5.2,\ConHig)			{$a_3$};
    \node[state, minimum size = 5mm]		(Am) at (\ConWid-2.2,\ConHig)	{$a_m$};
    \node[state, minimum size = 5mm]		(Am+) at (\ConWid-3,\ConHig+1)	{$a_{m+1}$};
    
    \node			at ($(L12)!0.5!(L1n)$)	{$\vdots$};
    \node			at ($(R12)!0.5!(R1n)$)	{$\vdots$};
    \node			at ($(Lm2)!0.5!(Lmn)$)	{$\vdots$};
    \node			at ($(Rm2)!0.5!(Rmn)$)	{$\vdots$};
    \node			at ($(R11)!0.5!(Lm1)$)	{$\dots$};
    \node			at ($(R12)!0.5!(Lm2)$)	{$\dots$};
    \node			at ($(R1n)!0.5!(Lmn)$)	{$\dots$};
    \node			at ($(G2)!0.5!(Gm)$)	{$\dots$};
    
    \path[line width = 0.75pt]
    		(L11)	edge						node	{}	(R11)
    				edge						node	{}	(L12)
    				edge						node	{}	(A1)
    		(L12)	edge						node	{}	(R12)
    				edge						node	{}	($(L12)!0.6cm!(L1n)$)
    		(L1n)	edge						node	{}	(R1n)
    				edge						node	{}	($(L1n)!0.6cm!(L12)$)
    		(R11)	edge						node	{}	($(R11)!1cm!(Lm1)$)
    		(R12)	edge						node	{}	($(R12)!1cm!(Lm2)$)
    		(R1n)	edge						node	{}	($(R1n)!1cm!(Lmn)$)
    		(Lm1)	edge						node	{}	($(Lm1)!1cm!(R11)$)
    		(Lm2)	edge						node	{}	($(Lm2)!1cm!(R12)$)
    		(Lmn)	edge						node	{}	($(Lmn)!1cm!(R1n)$)
    		(Lm1)	edge						node	{}	(Rm1)
    		(Lm2)	edge						node	{}	(Rm2)
    		(Lmn)	edge						node	{}	(Rmn)
    		(Rm1)	edge[in = 170, out = 10]		node	{}	(L11)
    		(Rm2)	edge[in = 170, out = 10]		node	{}	(L12)
    		(Rmn)	edge[in = 190, out = -10]	node	{}	(L1n)
    		(G1)	edge						node	{}	(A2)
    				edge						node	{}	($(G1)!1.5cm!(L12)$)
    				edge						node	{}	($(G1)!1.5cm!(R12)$)
    		(G2)	edge						node	{}	(A3)
    				edge						node	{}	($(G2)!1.5cm!($(R12)!1cm!(Lm1)$)$)
    				edge						node	{}	($(G2)!1.5cm!($(R12)!3cm!(Lm1)$)$)
    		(Gm)	edge						node	{}	(Am)
    				edge						node	{}	(Am+)
    				edge						node	{}	($(Gm)!1.5cm!(Lm2)$)
    				edge						node	{}	($(Gm)!1.5cm!(Rm2)$)
    		(A1)	edge[in = 180, out = 90]		node	{}	(G1)
		(A2)	edge						node	{}	(G2)
		(A3)	edge						node	{}	($(A3)!0.5cm!(Am)$)
		(Am)	edge						node	{}	($(Am)!0.5cm!(A3)$);
    		
	\draw	(Am+) 			to[out = 175, in = 45] 	(-0.2, \ConHig);
	\draw	(-0.2, \ConHig) 	to[out = 225, in = 135] 	(L1n);    		
    \end{tikzpicture}
    \caption{Construction of the graph $\Gamma_\phi$.}
    \label{fig:3SATCon}
\end{figure}
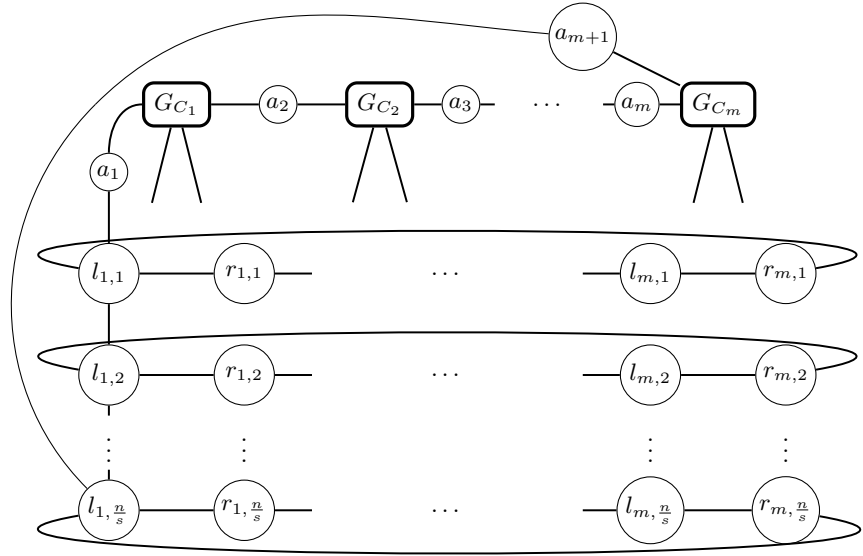

\begin{itemize}
    \item We start by creating vertices $l_{i,1}, \dots, l_{i, n/s}$ and $r_{i,1}, \dots, r_{i, n/s}$ for $i = 1, \dots, m$ and some constant $s$ to be determined later\footnote{If $n$ is not divisible by $s$, we may either add dummy variables until it is, or lower the demand of $l_{i,n/s}$ and $r_{i,n/s}$.}.
    \item We set the demand of $l_{1,j}$ to $2^s + 1$ for $j = 1, \dots, n/s$ and add edges $\{l_{1,j}, l_{1,j+1} \}$ for $j = 1, \dots, n/s - 1$.
    \item We set the demand of every other $l_{i,j}$ and every $r_{i,j}$ to $2^s$ and add edges $\{ l_{i,j}, r_{i,j} \}$, $\{ r_{i,j}, l_{i+1,j} \}$ and $\{r_{m,j}, l_{1,j} \}$ for $i = 1, \dots, m-1$ and $j = 1, \dots, n/s$.
    \item We connect $l_{1,1}$ to $l_{1,n/s}$ using a path $a_1, \dots, a_{m+1}$.
    \item For $i = 1, \dots, m$ let $x_a, x_b, x_c$ be the variables appearing in $C_i$. We set $j_1 = \lceil a/s \rceil, j_2 = \lceil b/s \rceil, j_3 = \lceil c/s \rceil$. Let
    \[
    X = \{l_{i,j_1}, l_{i,j_2}, l_{i,j_3}, r_{i,j_1}, r_{i,j_2}, r_{i,j_3}\}
    \]
    and let $\mathcal{F}_{C_i}$ be the set of all 
    \[
    F = \{\{l_{i,j_1},r_{i,j_1}\}^{q_1}, \{l_{i,j_2},r_{i,j_2}\}^{q_2}, \{l_{i,j_3},r_{i,j_3}\}^{q_3}, \}
    \]
    such that $Q = \{ (q_1, j_1), (q_2, j_2), (q_3, j_3) \}$ satisfies $C_i$.
    \item For $i = 1, \dots, m$ we implement a scanner gadget $G_{C_i}$ using the tuple \break $(X_i, a_i, a_{i+1}, \mathcal{F}_{C_i})$
   \end{itemize}

We prove the following useful facts about this graph.

\begin{lemma} \label{lem:EqGam}
$\Gamma_\phi$ is a yes instance of \textsc{MVTSP} if and only if $\phi$ has a satisfying assignment.
\end{lemma}

\begin{proof}
Let $x_1, \dots, x_n$ be the variables used in the formula $\phi$. Let $\chi_1, \dots, \chi_n$ be some satisfying assignment. We first define the tour on the construction before implementing the scanner gadgets, which we will refer to as $\Gamma_\phi'$, and then use Lemma~\ref{lem:ConGad} to find the desired tour on $\Gamma_\phi$. Set $c(l_{i,j}, l_{i+1,j}) = c(l_{1,1}, a_1) =  c(a_{m+1}, l_{1,n/s}) = c(a_i, a_{i+1}) = 1$. We choose
\[
c'(l_{i,j}, r_{i,j}) = 1 + \sum \limits _{k=1}^{s} 2^{k-1} \chi_{(j-1)s + k}
\]
for $i = 1, \dots, m$ and $j = 1, \dots, n/s$. Due to the chosen demands we need to choose
\[
c'(r_{i,j}, l_{i+1,j}) = 2^{s+1} - c'(l_{i,j}, r_{i,j})
\]
for $i = 1, \dots, m$ and $j = 1, \dots, n/s$, where we interpret $i$ modulo $m$, i.e. $m + 1 \equiv 1$. Note that $c'$ is connected and satisfies the demands on $\Gamma_\phi'$. Also note that since $\chi$ is a satisfying assignment, $c'$ is consistent with all the scanner gadgets $G_{C_i}$ and thus by Lemma~\ref{lem:ConGad} we there is some valid tour $c$ on $\Gamma_\phi$.

Now suppose we find a valid tour $c$ on $\Gamma_\phi$. Then by Lemma~\ref{lem:ConGad} there exists a tour $c'$ on $\Gamma_\phi'$ consistent with each gadget $G_{C_i}$. By definition of $G_{C_i}$ the values of $c'(l_{i,j}, r_{i,j})$ encode an assignment satisfying $C_i$ for $i = 1, \dots, m$. Since for $i \geq 2$ the demands of $l_{i,j}$ and $r_{i,j}$ equal $2^s$ we have that $c'(l_{i,j}, r_{i,j}) = 2^{s+1} - c'(r_{i,j}, l_{i+1, j}) = c'(l_{i+1,j}, r_{i+1,j})$ and therefore the values of $c'(l_{1,j}, r_{1,j})$ encode an assignment satisfying all clauses $C_1,\ldots,C_m$, which means we find an assignment which satisfies $\phi$.
\end{proof}

\begin{lemma} \label{lem:PWGam}
$\Gamma_\phi$ has pathwidth at most $3n/s + 21$.
\end{lemma}

\begin{proof}
We define the bags of the decomposition as follows. First we add
\[
A = \{ l_{1,1}, \dots, l_{1, n/s} \}
\]
to every bag. Let $W_1, \dots, W_{l_i}$ be a path decomposition of $G_{C_i}$. We define bag $X_{i,j}$ as follows
\[
X_{i,j} = A \cup \{l_{i,k}\}_{k=1}^{n/s} \cup \{r_{i,k}\}_{k=1}^{n/s} \cup W_j.
\]
We then define $Y_i$ as $\{l_{i+1,k}\}_{k=1}^{n/s} \cup \{r_{i,k}\}_{k=1}^{n/s}$ The final path decomposition then becomes
\[
X_{1,1}, \dots, X_{1,l_1}, Y_1, X_{2,1}, \dots, X_{i,l_i},Y_i, X_{i+1,1}, \dots, X_{m,l_m}.
\]
Note that all vertices and edges are covered by the decomposition. The set of bags containing any of the vertices of $A$ gives the whole decomposition. The set of bags containing any $l_{i,j}$ or $r_{i,j}$ for $i \geq 2$ gives the path $X_{i,1} \dots X_{i,l_1}$ with $Y_i$ at the end for $r_{i,j}$ and $Y_{i-1}$ at the beginning for $l_{i,j}$. Any vertex in the gadgets gives a single set $X_{i,j}$. By Lemma~\ref{lem:GadPW} the width of this path decomposition is at most\footnote{We don't include the term $|X|$, since $X \subseteq \{l_{i,k}\}_{k=1}^{n/s} \cup \{r_{i,k}\}_{k=1}^{n/s}$.}
\[
3\frac{n}{s} + 21.
\]
\end{proof}

Now we use our reduction to prove the following lower bound:

%\begingroup
%\def\thetheorem{\ref{thm:3SatRed}}
\begin{theorem}
Let $M$ be an upper bound on the demands in a graph $G$. Then \textsc{MVTSP} cannot be solved in time $f(\pw)\min\{|V(G)|,M\}^{o(\pw)}|V(G)|^{\Oh(1)}$, unless ETH fails.
\end{theorem}
%\addtocounter{theorem}{-1}
%\endgroup

\begin{proof}
We start by proving the following claim.
\begin{numclaim}\label{claim:GadSiz}
$|V(G_{C_i})| = \Oh(2^{3s})$ for $i = 1, \dots, m$.
\end{numclaim}
\begin{subproof}
Note that $\mathcal{F}_{C_i}$ is defined on at most three unique edges with each edge being chosen at most $2^s$ times\footnote{Due to the way we interpret the multiplicities as truth assignments (in particular the `$-1$') we know each edge gets chosen at least once.}. Therefore we can represent $\mathcal{F}_{C_i}$ by tuples $(z_1, z_2, z_3) \in [2^s]^3$. Since each tuple contributes a path of $z_1 + z_2 + z_3$ vertices, we find that
\begin{align*}
    |V(G_{C_i})| &= 8 + |\mathcal{F}_{C_i}| + \sum \limits _{(z_1, z_2, z_3) \in \mathcal{F}_{C_i}} z_1 + z_2 + z_3 \\
                &\leq 2^{3s + 1} + \sum \limits _{z_1, z_2 = 1}^{2^s} \left( 2^s (z_1 + z_2) + \sum \limits _{z_3 = 1}^{2^s} z_3 \right) \\
                &\leq 2^{3s + 1} + \sum \limits _{z_1, z_2 = 1}^{2^s} \left( 2^s (z_1 + z_2) + 2^{s+1} \right) \\
                &\leq 2^{3s + 1} + \sum \limits _{z_1 = 1}^{2^s} \left( 2^{2s} z_1 + 2^{2s+1} + 2^{2s + 1} \right) \\
                &\leq 2^{3s+3}.
\end{align*}\end{subproof}

Note that by Lemma~\ref{lem:EqGam}, solving a \textsc{3-CNF-SAT} instance $\phi$ reduces to solving \textsc{MVTSP} on $\Gamma_\phi$ for some choice of $s$. We remark that \[\Oh(f(\pw)\min\{|V(G)|,M\}^{o(\pw)}|V(G)|^{\Oh(1)}) \le \Oh(f(\pw)M^{o(\pw)}|V(G)|^{\Oh(1)}).\] It is therefore sufficient to show that there is no $\Oh(f(\pw) M^{o(\pw)} |V(G)|^{\Oh(1)})$ time algorithm for \textsc{MVTSP}, unless ETH fails.

Suppose we have a $\Oh\left(f(\pw) M^{o(\pw)} |V(G)|^{\Oh(1)}\right)$ time algorithm for \textsc{MVTSP}. Let $s = 4n/g(n)$ for some strictly increasing function $g(n) = 2^{o(n)}$ such that $f(g(n)) = 2^{o(n)}$. Note that $s = o(n)$ and $\pw \leq g(n)$ for large enough $n$. We construct the instance $\Gamma_\phi$ as previously described. We first note that by claim \ref{claim:GadSiz}
\[
|V(\Gamma_\phi)| = 2m\frac{n}{s} + \sum_{i=1}^m |V(G_{C_i})| = \Oh\left(m\left(\frac{n}{s} + 2^{3s}\right)\right)
\]
and by Lemma~\ref{lem:PWGam} we have that for any choice of $s$ and large enough $n$, $\Gamma_\phi$ has pathwidth at most $4n/s$. By applying our hypothetical algorithm for \textsc{MVTSP} to $\Gamma_\phi$ we now find an algorithm for \textsc{3-CNF-SAT} running in time
\begin{align*}
\Oh\left (f(\pw) M^{o(\pw)} |V(\Gamma_\phi)|^{\Oh(1)} \right) 	&= \Oh \left(f(4n/s) (2^s)^{o(n/s)} \left( m \left( \frac{n}{s} + 2^{3s}\right)\right)^{\Oh(1)} \right) \\
								&= \Oh\left(f(g(n)) \cdot 2^{o(n)} \cdot \left(m \left( g(n)/4 + 2^{o(n)} \right) \right)^{\Oh(1)} \right).
\end{align*}
We may assume that $m = 2^{o(n)}$ by the sparsification lemma. Using this and the fact that $g(n) = 2^{o(n)}$ we find
\begin{align*}
								&= \Oh \left(2^{o(n)} \cdot  \left(2^{o(n)} \right)^{\Oh(1)} \right) \\
								&= \Oh\left(2^{o(n)}\right).
\end{align*}
This contradicts ETH, completing our proof.
\end{proof}
\section{Conclusion and Further Research}\label{sec:conc}

We initiated the study of the parameterized complexity of the \textsc{Connected Flow} problem and showed that the problem behaves very differently when parameterized by the number of demand vertices, the size of the vertex cover of the graph, or treewidth of the input graph.

While we essentially settled the complexity of the variants of the problem parameterized by the number of demands or by the treewidth, we still leave the following questions open for the vertex cover parameterization: 

Can \textsc{Connected Flow} be solved in $\Os(c^{\Oh(k)})$ time, with $c$ a constant and $k$ the size of the vertex cover of the input graph? Such an algorithm would be a strong generalization of the algorithms from~\cite{BergerKMV20,kowalik_et_al:LIPIcs:2020:12932}. While we believe our approach from Theorem~\ref{thm:XFPT} makes significant progress towards solving this question affirmatively, it seems that non-trivial ideas are required.

Does \textsc{Connected Flow} admit a kernel polynomial in $k$ where $k$ is the size of the vertex cover if the input graph? It seems that especially the capacities can make the problem a lot harder. It would be interesting to see if our arguments for Theorem~\ref{thm:Xkernel} can be extended to kernelize this more general problem as well.

\bibliographystyle{abbrv}
\bibliography{bib}
\clearpage
\appendix
\section{Problem Definitions}\label{sec:problemdefs}\label{app:Definitions}

In this section we formally introduce and discuss a number of computational problems that are relevant for this paper.

Formally, we define the \textsc{Flow} problem as follows.

\defproblem{ \textsc{Flow}}{Given digraph $G = (V,E)$, $D \subseteq V$, $\Dem: D\to \N$, $\Dis: E \to \N$, $\Ca: E \to \N\cup\{\infty \}$}{Find a function $f : E \to \N$ such that
	\begin{itemize}		\setlength\itemsep{0em}
		%\item for every $u,v\in V$ such that $(u,v)\not\in E$, $f(u,v)=0$ %\todo{K: added} C: now added it in the definition of f itself :)
		%\item $G_f$ is connected, 
		\item for every $v \in V$ we have $\sum_{u \in V} f(u,v)  = \sum_{u \in V} f(v,u)$,
		\item for every $v \in D$ we have 
		$\sum_{u \in V} f(u,v) = \Dem(v)$,
		\item for every $e \in E: f(e) \le \Ca(e)$,
	\end{itemize}
		and the value $\Dis(f) = \sum_{e \in E} \Dis(e)f(e)$ is minimized.}\label{def:conflow}

From the definition it is clear that apart from the connectivity requirement, it is indeed equivalent to \textsc{Connected Flow}.

We will use the following standard definition of \textsc{Min Cost Flow}.

\defproblem{\textsc{Min Cost Flow}}{Digraph $G=(V,E)$ with source node set $S\subseteq V$ and sink nodes $T\subseteq V$, $\Dis:E\to \N$, $\Ca: E\to \N\cup\infty$}{Find a function $f:E\to \N$ such that
	\begin{itemize} \setlength\itemsep{0em}
		\item for every $v \in V\setminus ( T \cup S)$ we have $\sum_{u \in V} f(u,v)  = \sum_{u \in V} f(v,u)$,
		\item for every $e \in E: f(e) \le \Ca(e)$,
		\item the value of $\sum_{v\in S} \sum_{u\in V} f(v,u)$ is maximal,
	\end{itemize}
	and the value $\Dis(f) = \sum_{e \in E} \Dis(e)f(e)$ is minimized.
}

\paragraph*{Equivalence of \textsc{Flow} and \textsc{Min Cost Flow}.} We argue that \textsc{Flow} is equivalent to \textsc{Min Cost Flow} by simple reductions. First we reduce in the forward way. For each $d \in D$, create vertices $d_{\Out}, d_{\In}$ where $d_{\Out}$ is a source node with outgoing flow $\Dem(d)$ and $d_{\In}$ is a sink node with ingoing flow $\Dem(d)$. For all other vertices in $V\setminus D$, create a node and connect to all its neighbors, where all outgoing edges to a vertex in $D$ go to $d_{\In}$ and all ingoing edges from a vertex in $D$ connect to $d_{\Out}$. 

For the other way, let $S$ be the set of source nodes and $T$ be the set of sink nodes of the \textsc{Min Cost Max Flow} problem. Then add one `big' node $x$ to the graph, with demand equal to the outgoing flow from all the source nodes. Then add $(t,x)$ for all $t \in T$ with $\Dis(t,x) = 0$, $\Ca(t,x) = \Out(t)$. Furthermore add $(x,s)$ for all $s \in S$ with $\Dis(x,s) = 0$, $\Ca(x,s) =  \In(s)$. 

Since \textsc{Min Cost Flow} is well-known to be solvable in polynomial time, we can therefore conclude that \textsc{Flow} is solvable in polynomial time as well.

In Kowalik et al.~\cite{kowalik_et_al:LIPIcs:2020:12932}, the Many Visit TSP (MVTSP) is defined as follows. 

\defproblem{\textsc{Many Visits TSP (MVTSP)}}{Digraph $G=(V,E)$, $\Dem: V\to \N$, $\Dis: V^2 \to \N$}{Find a minimal cost tour $c$, such that each $v\in V$ is visited exactly $\Dem(v)$ times.}	

Note that MVTSP is a special case of \textsc{Connected Flow}, where $D=V$ and the capacities of all edges are infinite.

\end{document}